\documentclass[preprint,11pt]{elsarticle}
\usepackage{amsmath,amsthm,amsfonts,amssymb,bm,graphicx,algorithm,hyperref,cleveref,subcaption,booktabs}
\usepackage[left=20mm,right=20mm,top=20mm,bottom=20mm]{geometry}
\usepackage[noend]{algorithmic}

\newtheorem{theorem}{Theorem}[section]
\newtheorem{proposition}{Proposition}[section]

\newtheorem{remark}{Remark}[section]
\newtheorem{corollary}{Corollary}[section]
\usepackage{xcolor}

\numberwithin{equation}{section}

\crefname{equation}{}{}
\crefname{algorithm}{Algorithm}{}
\crefname{theorem}{Theorem}{}
\crefname{remark}{Remark}{}
\crefname{figure}{Figure}{}

\begin{document}

\title{Frozen Gaussian sampling algorithms for simulating Markovian open quantum systems in the semiclassical regime}

\author[1]{Limin Xu}
\ead{xulimin@westlake.edu.cn}

\author[2]{Zhen Huang}
\ead{hertz@math.berkeley.edu}

\author[3]{Zhennan Zhou\corref{cor1}}
\ead{zhouzhennan@westlake.edu.cn}

\cortext[cor1]{Corresponding author}

\address[1]{Institute for Theoretical Sciences, Westlake University, Hangzhou, 310030, China}
\address[2]{Department of Mathematics, University of California, Berkeley, California 94720 USA}
\address[3]{Institute for Theoretical Sciences, Westlake University, Hangzhou, 310030, China}

\begin{abstract}
Simulating Markovian open quantum systems in the semiclassical regime poses a grand challenge for computational physics, as the highly oscillatory nature of the dynamics imposes prohibitive resolution requirements on traditional grid-based methods. To overcome this barrier, this paper introduces an efficient Frozen Gaussian Sampling (FGS) algorithm based on the Wigner-Fokker-Planck phase-space formulation. The proposed algorithm exhibits two transformative advantages. First, for the computation of physical observables, its sampling error is independent of the semiclassical parameter $\varepsilon$, thus fundamentally breaking the prohibitive computational scaling faced by grid methods in the semiclassical limit. Second, its mesh-free nature entirely eliminates the boundary-induced instabilities that constrain long-time grid-based simulations. Leveraging these capabilities, the FGS algorithm serves as a powerful investigatory tool for exploring the long-time behavior of open quantum systems. Specifically, we provide compelling numerical evidence for the existence of steady states in strongly non-harmonic potentials—a regime where rigorous analytical results are currently lacking.
\end{abstract}


\maketitle

\section{Introduction}
\label{sec:introduction}
In the foundational narrative of quantum mechanics, physical systems are often idealized as closed entities, evolving unitarily in isolation from the universe. Reality, however, is far more intricate; most realistic scenarios involve quantum systems that inevitably interact with an external environment, leading to complex, non-unitary phenomena such as dissipation and decoherence. Such systems, known as open quantum systems (OQS), have garnered significant interest in recent years, a surge driven by experimental advancements that enable the controlled engineering of dissipation channels \cite{diehl2008quantum}. These developments have not only deepened our fundamental understanding but also led to the observation of novel out-of-equilibrium behaviors not typically seen in their closed counterparts \cite{breuer2002theory}.

The cornerstone for describing the dynamics of Markovian open systems is the Gorini-Kossakowski-Sudarshan-Lindblad (GKSL) master equation, commonly referred to as the Lindblad equation \cite{gorini1976completely,lindblad1976generators}. Its robust physical and mathematical foundations have established it as an indispensable tool across a vast spectrum of disciplines, including chemical physics \cite{cao2017lindblad,wang2020combining}, condensed matter physics \cite{buvca2012note,lotem2020renormalized}, and quantum computing \cite{chen2023quantum,chen2025randomized}. Despite its ubiquity, a formidable challenge emerges when attempting to simulate these dynamics in the semiclassical regime ($\varepsilon \ll 1$). In this limit, the quantum state becomes highly oscillatory in both space and time on a scale of $O(\varepsilon)$, imposing prohibitive resolution requirements on traditional numerical methods. Approaches like finite differences or finite elements, for instance, would demand spatial discretization with a mesh size scaling as $O(1/\varepsilon^{2m})$ (where $m$ is the physical dimension), rendering such simulations computationally intractable as $\varepsilon\rightarrow 0$. Overcoming this computational barrier is thus a critical prerequisite for accurately modeling a wide array of important quantum phenomena.

A promising route to circumvent this computational barrier lies in shifting the perspective from Hilbert space to phase space. An equivalent formulation of the Lindblad dynamics is provided by the Wigner-Fokker-Planck (WFP) equation, obtained via the Wigner transformation of the density matrix:
\begin{equation}
    W(\bm{x},\bm{\xi},t)=\frac{1}{(\pi\varepsilon)^n}\int\rho(\bm{x}+\bm{y},\bm{x}-\bm{y}, t) e^{-\frac{2\mathrm{i}}{\varepsilon}\bm{y}\cdot\bm{\xi}}\mathrm{d}\bm{y}.
    \label{eq:Wigner_transform_intro}
\end{equation}
This phase space framework is particularly well-suited for semiclassical analysis, as it offers a direct and intuitive connection to classical mechanics, formally reducing to a classical Fokker-Planck equation in the $\varepsilon \to 0$ limit \cite{graefe2018lindblad}. More profoundly, the rationale for applying wavepacket-based methods to open systems is substantiated by recent breakthroughs in the fundamental understanding of their semiclassical dynamics. A striking feature of OQS is that their Ehrenfest time—the timescale for which quantum evolution can be approximated by classical trajectories—is significantly longer than that of their closed counterparts. This remarkable phenomenon has been recently proven rigorously \cite{hernandez2025classical,galkowski2025classical,li2024long} and verified numerically \cite{galkowski2025classical}. Crucially, the mathematical strategy underpinning these proofs involves decomposing the initial state into a sum of Gaussian states and analyzing their subsequent evolution. This theoretical insight strongly suggests that numerical methods based on Gaussian wavepackets, such as the Frozen Gaussian Approximation (FGA) and its sampling variants, which have proven highly successful for the Schr\"odinger equation \cite{lu2018frozen,huang2023efficient} and high-frequency wave propagation \cite{lu2012convergence}, represent a natural framework for simulating OQS.

Motivated by this potent theoretical direction, this paper presents the formal development and implementation of a Frozen Gaussian Sampling (FGS) algorithm, specifically tailored to solve the Wigner-Fokker-Planck equation in the semiclassical regime. Our contribution extends beyond a mere application of existing semiclassical ideas; we establish the complete theoretical foundation required to accurately propagate Gaussian wavepackets in an open quantum system. Central to our work is the rigorous derivation of the full system of ordinary differential equations that governs the evolution of each wavepacket's key parameters: the trajectory of its center in phase space, the dynamic reshaping of its covariance matrix, and the non-trivial evolution of its amplitude under dissipative influence. Furthermore, to ensure the validity and stability of the entire framework, we provide a crucial mathematical proof demonstrating that the covariance matrix remains positive definite throughout the time evolution. This pivotal result guarantees that the Gaussian ansatz is always well-defined and physically meaningful, forming the bedrock of a robust and reliable numerical algorithm.

The algorithm grounded in this rigorous framework possesses transformative advantages over traditional grid-based approaches. First, it fundamentally breaks the unfavorable scaling relationship between computational cost and the semiclassical parameter. For the computation of physical observables, its sampling error is demonstrated to be nearly independent of $\varepsilon$, thus overcoming the severe resolution requirements that render grid methods prohibitively expensive in the deep semiclassical limit. Second, as a mesh-free method that evolves an ensemble of trajectories in an unbounded phase space, it is inherently robust for long-time simulations. This completely avoids the boundary-induced instabilities and spurious artifacts that inevitably constrain and corrupt fixed-grid simulations as the Wigner function expands and explores a large domain. Empowered by this combination of efficiency and robustness, our FGS algorithm serves not just as a numerical solver, but as a powerful investigatory tool for exploring complex physical phenomena. In this work, we leverage its capabilities to provide compelling numerical evidence for the existence of a unique steady state in strongly non-harmonic potentials—a regime where a rigorous analytical proof remains an open problem.

The remainder of this article is organized as follows. In Section \ref{sec:problem_setup}, we establish the theoretical background by detailing the Wigner-Fokker-Planck equation and reviewing the wavepacket dynamics for closed systems. Section \ref{sec:wavepacket_dynamics} presents our core theoretical results, including the derivation of the wavepacket dynamics for the WFP equation and the crucial proof that the covariance matrix remains positive definite, a property necessary for the algorithm to be well defined. The complete algorithm is constructed in Section \ref{sec:algorithm}, followed by a formal analysis of its sampling error. In Section \ref{sec:experiments}, an extensive set of numerical experiments is presented to validate the algorithm's performance and to apply it to the physical phenomena mentioned above. Finally, Section \ref{sec:conclusions} concludes the paper with a summary of our findings and a discussion of future research directions.

\section{Problem setup}\label{sec:problem_setup}
\subsection{Lindblad operator and the Wigner Fokker-Planck equation}
We consider a quantum system in an $n$-dimensional configuration space with Hilbert space $\mathcal{H} = L^2(\mathbb{R}^n)$, where the position vector is denoted by $\bm{x}=(x_1,x_2,\cdots,x_n)^{T}$. The evolution of the density matrix $\hat{\rho}$ for such a system is described by the Gorini-Kossakowski-Sudarshan-Lindblad (GKSL) equation:
\begin{equation}
    \begin{aligned}
        \partial_t\hat\rho = -\frac{\mathrm i}{\varepsilon} \left[\hat H,\hat\rho\right] + \frac{1}{\varepsilon}\sum_{j\in J} \left(\hat L_j\hat\rho\hat L_j^\dagger - \frac 1 2\left\{\hat L_j^\dagger\hat L_j,\hat\rho\right\}\right).
    \end{aligned} 
    \label{eq: Lindblad equation in operator}
\end{equation}
Here, the parameter $\varepsilon$ is the rescaled Planck constant, The operator $\hat{H}$ is the Lindblad corrected Hamiltonian, which consists of the system Hamiltonian $\hat{H}_{\text{S}}$ and a Lamb shift term $\hat{H}_{\text{lamb}}$ describing environmental renormalization effects \cite{manzano2020short}:
\begin{equation}
    \hat{H} = \hat{H}_{\text{S}} + \hat{H}_{\text{lamb}}, \quad 
    \hat{H}_{\text{S}} = -\frac{\varepsilon^2}{2}\Delta + V(\bm{x}), \quad 
    \hat{H}_{\text{lamb}} = \sum_{k=1}^n \frac{\mu_k}{2}\{\hat{x}_k, \hat{p}_k\}, \quad \mu_k \in \mathbb{R}.
\end{equation}
Here, $V(\bm{x})$ is the external potential, $\mu_k$ are the renormalization strengths, and $\hat{x}_k$ and $\hat{p}_k = -\mathrm{i}\varepsilon \partial_{x_k}$ are the position and momentum operators along the $k$-th direction ($k=1,\cdots,n$), respectively. 

The operators $\{\hat{L}_j\}$ are the jump operators that model the system-bath coupling, and the coefficients $\gamma_j > 0$ are the corresponding coupling strengths. For a detailed physical background on these terms, we refer the reader to \cite{dubois2021semi,graefe2018lindblad}. A physically important case, which we adopt in this work, is when the jump operators are linear combinations of the position and momentum operators:
\begin{equation}
\hat L_{k\nu} = a_{k\nu} \hat{x}_k + b_{k\nu} \hat{p}_k,\quad 
k = 1,\cdots,n, \quad \nu= 1,\cdots, n_{\nu},
\end{equation}	
where $a_{k\nu},b_{k\nu}\in\mathbb{C}$. The index set $J$ for jump operators is thus $J = \{(k\nu)|k=1,\cdots,n,\nu=1,\cdots,n_\nu\}$.

By writing the operators in \cref{eq: Lindblad equation in operator} in the coordinate representation, we obtain the following partial differential equation for the density matrix kernel $\rho(\bm{x},\bm{y},t)$:
\begin{equation}
    \partial_t\rho(\bm{x},\bm{y},t) = -\frac{\mathrm{i}}{\varepsilon}\left(-\frac{\varepsilon^2}{2}\Delta_{\bm{x}}+V(\bm{x})+\frac{\varepsilon^2}{2}\Delta_{\bm{y}}-V(\bm{y})\right)\rho + D(\rho),
    \label{eq:Lindblad_coordinate_representation}
\end{equation}
where the term $D(\rho)$ characterizes the environmental effects, including drift, diffusion, dissipation, and renormalization (Lamb shift). It is composed of $D(\rho) = D_1(\rho)+D_2(\rho)+D_3(\rho)$, with
\begin{align}
    D_1(\rho) &= -\frac{1}{2\varepsilon}\sum_{k=1}^n \left( \alpha_k (x_k-y_k)^2\rho - \beta_k (\varepsilon \partial_{x_k}+\varepsilon\partial_{y_k})^2 \rho \right), \\
    D_2(\rho) &= \sum_{k=1}^{n} \left( -\operatorname{Im}(\gamma_k) +\mathrm{i}\left[ (\gamma_k x_k\partial_{y_k}-\overline{\gamma}_ky_k\partial_{x_k}) + \operatorname{Re}(\gamma_k)(x_k\partial_{x_k}-y_k\partial_{y_k}) \right] \right)\rho, \\
    D_3(\rho) &= -\sum_{k=1}^n \mu_k (x_k\partial_{x_k}+y_k\partial_{y_k}+1)\rho.
\end{align}

The coefficients $\alpha_k, \beta_k, \gamma_k$ are determined by the jump operators:
\begin{equation}
    \alpha_k = \sum_{\nu=1}^{n_\nu} |a_{k\nu}|^2, \quad 
    \beta_k = \sum_{\nu=1}^{n_\nu} |b_{k\nu}|^2, \quad 
    \gamma_k = \sum_{\nu=1}^{n_\nu} a_{k\nu} \bar{b}_{k\nu}.
    \label{eq:dissipative_coeffs}
\end{equation}

The key step in moving to a phase space formulation is to apply the Wigner transformation (defined in \cref{eq:Wigner_transform_intro}) to the Lindblad equation in its coordinate representation \eqref{eq:Lindblad_coordinate_representation}. The detailed derivation is provided in \ref{appendix:wfp_derivation}. This procedure yields the Wigner-Fokker-Planck (WFP) equation:
\begin{equation}
    \partial_t W(\bm{x}, \bm{\xi}, t) + \bm{\xi} \cdot \nabla_{\bm{x}} W + \Theta[V]W = \mathcal{D}(W),
    \label{eq:WFP_main}
\end{equation}
where $W(\bm{x}, \bm{\xi}, t)$ is the Wigner function. The potential term $\Theta[V]$ is a pseudo-differential operator given by
\begin{equation}
    \Theta[V] W(\bm{x}, \bm{\xi}, t) = \frac{1}{(\pi \varepsilon)^n} \iint \left[\frac{\mathrm{i}}{\varepsilon}(V(\bm{x}+\bm{y}) - V(\bm{x}-\bm{y}))\right] W(\bm{x}, \boldsymbol{\xi}', t) e^{\frac{2\mathrm{i}}{\varepsilon} \bm{y} \cdot (\boldsymbol{\xi}'-\bm{\xi})} \mathrm{d}\boldsymbol{\xi}' \mathrm{d}\bm{y},
    \label{eq:potential_operator}
\end{equation}
and the term $\mathcal{D}(W)$, which encapsulates all drift, diffusive and dissipative effects from the environment, is a second-order differential operator:
\begin{equation}
    \begin{aligned}
        \mathcal{D}(W) = & \frac{\varepsilon}{2} \sum_{k=1}^n \left(\alpha_k \partial_{\xi_k\xi_k}W + \beta_k \partial_{x_kx_k}W - 2\operatorname{Re}(\gamma_k) \partial_{x_k\xi_k}W\right) \\
        & - \sum_{k=1}^n \left[ (\operatorname{Im}(\gamma_k) + \mu_k) \partial_{x_k}(x_k W) + (-\operatorname{Im}(\gamma_k) + \mu_k) \partial_{\xi_k}(\xi_k W) \right].
    \end{aligned}
    \label{eq:dissipative_operator}
\end{equation}

While the Lindblad and WFP equations are mathematically equivalent, we choose to develop our numerical method based on the WFP equation \eqref{eq:WFP_main}. This choice is motivated by several key advantages of the phase space formulation. First, the mathematical properties of the WFP equation, such as its well-posedness \cite{arnold2004analysis,arnold2007wigner} and long-time behavior \cite{arnold2004analysis,arnold2012wigner}, are well-established in the literature. Second, the WFP framework provides an intuitive connection to classical mechanics; as $\varepsilon \to 0$, the operator $\Theta[V]W$ formally converges to $-\nabla_{\bm{x}}V \cdot \nabla_{\bm{\xi}}W$, and the WFP equation reduces to a classical Fokker-Planck equation\cite{arnold2004analysis, sparber2004long} , offering a clear physical picture in the semiclassical limit. Finally, and most pertinently to our work, the phase space representation is particularly amenable to the wavepacket-based semiclassical approximations that form the foundation of our proposed algorithm.

\subsection{Review of Wavepacket Dynamics for the Wigner Equation}
\label{subsec:review_closed_system}
Before developing our method for the dissipative WFP equation, we first review the established wavepacket dynamics for its non-dissipative counterpart, the Wigner equation. This equation is recovered from \cref{eq:WFP_main} by setting all environmental parameters to zero:
\begin{equation}
    \partial_t W(\bm{x}, \bm{\xi}, t) + \bm{\xi} \cdot \nabla_{\bm{x}}W + \Theta[V]W = 0.
    \label{eq:Wigner_equation}
\end{equation}
This review lays the groundwork for our extension to open systems. The central idea of these methods is to approximate the Wigner function as a superposition of time-evolving Gaussian wavepackets:
\begin{equation}
    W(\bm{x}, \bm{\xi}, t) \approx \sum_{j} A_j \exp\left(-\frac{1}{2\varepsilon} T(\bm{x}, \bm{\xi}, \bm{q}_j(t), \bm{p}_j(t); \bm{\Sigma}_j(t))\right).
    \label{eq:Gaussian_sum_ansatz_time_dependent}
\end{equation}
Here, $T$ is a general quadratic form defined as a function of the phase space variables $(\bm{x}, \bm{\xi})$, the wavepacket center $(\bm{q}, \bm{p})$, and the covariance matrix $\bm{\Sigma}$:
\begin{equation}
    T(\bm{x}, \bm{\xi}, \bm{q}, \bm{p}; \bm{\Sigma}) = \begin{pmatrix} \bm{x}-\bm{q} \\ \bm{\xi}-\bm{p} \end{pmatrix}^T \bm{\Sigma}^{-1} \begin{pmatrix} \bm{x}-\bm{q} \\ \bm{\xi}-\bm{p} \end{pmatrix}.
    \label{eq:quadratic_form_T_general}
\end{equation}
The expression \eqref{eq:Gaussian_sum_ansatz_time_dependent} represents a linear superposition of Gaussian wavepackets. For each wavepacket indexed by $j$, $A_j$ is its weight, $(\bm{q}_j(t), \bm{p}_j(t))$ is its time-dependent center, and $\bm{\Sigma}_j(t)$ is its time-dependent rescaled covariance matrix. For convenience, we also define the inverse covariance matrix $\bm{G}_j(t) = \bm{\Sigma}_j^{-1}(t)$.

The equations of motion for the parameters are derived via an asymptotic analysis. This procedure involves (1) approximating the potential $V(\bm{x})$ with a local quadratic function via a Taylor expansion around the wavepacket's position center, $\bm{q}(t)$:
\begin{equation}
    V(\bm{x}) \approx V(\bm{q}) + \nabla_{\bm{x}}V(\bm{q}) \cdot (\bm{x}-\bm{q}) + \frac{1}{2}(\bm{x}-\bm{q})^T \nabla_{\bm{x}}^2V(\bm{q}) (\bm{x}-\bm{q}),
    \label{eq:V_taylor_generic}
\end{equation}
 (2) substituting the Gaussian ansatz into the resulting simplified Wigner equation, and (3) matching terms at the leading orders of $\varepsilon$. This analysis (detailed in, e.g., \cite{lu2018frozen}) yields separate ODEs for each parameter set.

The $O(1)$ terms govern the amplitude, which for the closed system is constant:
\begin{equation}
    \frac{\mathrm{d}A_j}{\mathrm{d}t} = 0.
    \label{eq:O_1_dA_dt}
\end{equation}
The $O(\varepsilon^{-1})$ terms, after further differentiation, yield the dynamics for the wavepacket center and shape. The center $\bm{z}_j = (\bm{q}_j^T, \bm{p}_j^T)^T$ evolves according to the classical Hamiltonian dynamics:
\begin{equation}
    \frac{\mathrm{d}}{\mathrm{d}t}\begin{pmatrix} \bm{q}_j(t) \\ \bm{p}_j(t) \end{pmatrix} = \begin{pmatrix} \bm{p}_j(t) \\ -\nabla_{\bm{x}}V(\bm{q}_j(t)) \end{pmatrix}.
    \label{eq:center_dynamics}
\end{equation}
The inverse covariance matrix $\bm{G}_j(t)$ evolves according to the differential Lyapunov equation:
\begin{equation}
    \frac{\mathrm{d} \bm{G}_j}{\mathrm{d} t} = -(\bm{G}_j \bm{C}_j + \bm{C}_j^T \bm{G}_j), \quad \text{where} \quad \bm{C}_j(t) = \begin{pmatrix} \bm{0} & \bm{I} \\ -\nabla_{\bm{x}}^2 V(\bm{q}_j(t)) & \bm{0} \end{pmatrix}.
    \label{eq:G_dynamics_closed}
\end{equation}

The derivation of the center dynamics requires the matrix $\bm{G}_j(t)$ to be invertible, which is guaranteed if it remains positive definite (assuming $\bm{G}_j(0)$ is positive definite). This crucial property follows directly from the property of the differential Lyapunov equation, which was stated in the following proposition.

\begin{proposition}[\cite{dieci1994positive}, Theorem~2.1]
\label{prop:positive_definiteness_of_X}
    The solution $\bm{X}(t)$ to the differential Lyapunov equation
    \begin{equation}
        \frac{d\bm{X}}{dt} = \bm{A}(t)\bm{X} + \bm{X}\bm{A}(t)^T + \bm{D}(t), \quad \bm{X}(0)=\bm{X}_0,
    \end{equation}
    exists and is symmetric for all $t>0$ if $\bm{X}_0$ is symmetric. Furthermore, if $\bm{X}_0$ is positive definite and $\bm{D}(t)$ is positive semidefinite for all $t \ge 0$, then $\bm{X}(t)$ remains positive definite for all $t>0$.
\end{proposition}

\section{Wavepacket dynamics for the Wigner-Fokker-Planck equation in the semiclassical regime}
\label{sec:wavepacket_dynamics}
In this section, we extend the semiclassical wavepacket dynamics from the closed-system Wigner equation to the Wigner-Fokker-Planck equation. We present the resulting equations of motion for the wavepacket parameters in the presence of dissipation and diffusion; the detailed asymptotic analysis is deferred to \ref{appendix: wavepacket dynamics} due to its complexity. Crucially, we also provide a rigorous proof that the covariance matrix remains positive definite throughout the evolution, a property that is fundamental to the stability of the proposed algorithm.

\subsection{Wavepacket Dynamics}
We now extend the semiclassical procedure to the  WFP equation. The Wigner function is again approximated by the Gaussian wavepackets ansatz \eqref{eq:Gaussian_sum_ansatz_time_dependent}, and for each generic wavepacket, the potential $V(\bm{x})$ is locally approximated by the quadratic form \eqref{eq:V_taylor_generic}. The starting point for our derivation is therefore the WFP equation under this local potential, which governs the evolution of a generic wavepacket $W_G$:
\begin{equation}
\label{eq:WFP_local_potential}
\begin{aligned}
&\partial_t W_G + \bm{\xi}\cdot\nabla_{\bm{x}}W_G - \left(\nabla_{\boldsymbol{x}}V(\boldsymbol{q})+\nabla_{\boldsymbol{x}}^2V(\boldsymbol{q})(\boldsymbol{x}-\boldsymbol{q})\right)\cdot\nabla_{\bm{\xi}}W_G \\
&= \frac{\varepsilon}{2}\sum_{k=1}^n\left(\alpha_k\partial_{\xi_k\xi_k}W_G+\beta_k\partial_{x_kx_k}W_G- 2\operatorname{Re}(\gamma_k) \partial_{x_k}\partial_{\xi_k}W_G\right)\\
&\quad - \sum_{k=1}^n \left[ (\operatorname{Im}(\gamma_k) + \mu_k)\partial_{x_k}(x_kW_G) +(-\operatorname{Im}(\gamma_k)  +\mu_k)\partial_{\xi_k} (\xi_k W_G) \right].
\end{aligned}
\end{equation}

Unlike the closed-system case, the presence of the diffusive and dissipative terms in \cref{eq:WFP_local_potential} prevents its reduction to a simple Liouville form. The derivation of the parameter dynamics is therefore more involved, requiring a direct substitution of the Gaussian ansatz into \cref{eq:WFP_local_potential} and a detailed matching of coefficients. As the full derivation is lengthy, it is provided in \ref{appendix: wavepacket dynamics}. We present the final equations of motion here.

For a compact representation, we first define the following auxiliary matrices. Let $\Gamma = \mathrm{diag}(\gamma_k)$ and $\mathcal{M} = \mathrm{diag}(\mu_k)$, then:
\begin{align}
    \Gamma_1 &= \begin{pmatrix}
	    \bm{0} & \operatorname{Re}(\Gamma) \\
	    \operatorname{Re}(\Gamma) & \bm{0}
    \end{pmatrix}, &
    \Gamma_2 &= \begin{pmatrix}
	    -\operatorname{Im}(\Gamma) & \bm{0} \\
	    \bm{0} & -\operatorname{Im}(\Gamma)
    \end{pmatrix}, \label{eq:Gamma_matrices} \\
    \bm{B} &= \begin{pmatrix}
	    \mathrm{diag}(\beta_k) & \bm{0} \\
	    \bm{0} & \mathrm{diag}(\alpha_k)
	\end{pmatrix}, &
    \widetilde{\mathcal{M}} &= \begin{pmatrix}
	    -\mathcal{M} & \bm{0} \\
	    \bm{0} & \mathcal{M}
	\end{pmatrix}, \\
    \bm{C}(t) &= \begin{pmatrix}
	    \bm{0} & \bm{I}_n \\
	    -\nabla_{\bm{x}}^2 V(\bm{q}(t)) & \bm{0}
    \end{pmatrix}. \label{eq:C_matrix_def}
\end{align}

The asymptotic analysis yields the following complete set of ODEs for the generic wavepacket parameters:
\begin{align}
    \frac{\mathrm{d} \bm{G}}{\mathrm{d} t} &= (\Gamma_2+\widetilde{\mathcal{M}}-\bm{C}^T)\bm{G}+\bm{G}(\Gamma_2+\widetilde{\mathcal{M}}-\bm{C}) + \bm{G}(\Gamma_1-\bm{B})\bm{G}, \label{eq:G_dynamics_open} \\
    \frac{\mathrm{d}A}{\mathrm{d}t} &= \operatorname{Tr}\left(\Gamma_2-\frac{1}{2}(\bm{B}-\Gamma_1)\bm{G}\right)A, \label{eq:A_dynamics_open} \\
    \frac{\mathrm{d}\bm{q}}{\mathrm{d}t} &= \bm{p} + (\operatorname{Im}\left(\Gamma\right) + \mathcal{M})\bm{q}, \label{eq:q_dynamics_open} \\
    \frac{\mathrm{d}\bm{p}}{\mathrm{d}t} &= -\nabla_{\bm{x}}V|_{\bm{x}=\bm{q}} + (\operatorname{Im}\left(\Gamma\right) - \mathcal{M})\bm{p}. \label{eq:p_dynamics_open}
\end{align}
These equations represent the central theoretical result of this section. Several key differences from the closed-system dynamics are immediately apparent. First, the amplitude $A(t)$ is no longer constant but evolves dynamically according to \cref{eq:A_dynamics_open}. This signifies a dynamic redistribution of weight among the Gaussian basis packets, a hallmark of the non-unitary evolution in an open system. Second, the equations for the center $(\bm{q}, \bm{p})$ are modified by additional terms arising from the system-bath coupling. These terms render the trajectory non-Hamiltonian, causing it to deviate from the purely conservative flow seen in \cref{eq:center_dynamics}. The detailed structure of these dissipative trajectories warrants further investigation and will be a subject of future work.

Finally, for the Gaussian ansatz to remain valid, the matrix $\bm{G}(t)$ must be positive definite for all $t \ge 0$. We will rigorously prove this crucial property in the following subsection.

\subsection{Positive Definiteness and Stability}
In this subsection, we rigorously prove that the covariance matrix $\bm{\Sigma}(t)$ remains positive definite throughout the evolution. This property is fundamental to ensuring that the Gaussian ansatz is well-defined and that the numerical algorithm is stable.

\begin{theorem}
\label{thm:G_positive_definite}
Consider the system of ODEs \eqref{eq:G_dynamics_open}-\eqref{eq:p_dynamics_open}. If the initial matrix $\bm{G}(0)$ is symmetric and positive definite, then its solution $\bm{G}(t)$ remains symmetric and positive definite for all $t > 0$. And the corresponding covariance matrix $\bm{\Sigma}(t) = \bm{G}^{-1}(t)$ satisfies the following differential Lyapunov equation:
\begin{equation}
    \frac{\mathrm{d}\bm{\Sigma}}{\mathrm{d}t} = - \bm{\Sigma}(\Gamma_2+\mathcal{M}-\bm{C}^T) - (\Gamma_2+\mathcal{M}-\bm{C})\bm{\Sigma} - (\Gamma_1-\bm{B}).
    \label{eq:Sigma_dynamics_open}
\end{equation}
\end{theorem}

\begin{proof}
A direct proof of positive definiteness for $\bm{G}(t)$ from its governing equation \eqref{eq:G_dynamics_open} is challenging due to the equation's nonlinear (Riccati) nature. Therefore, we adopt an indirect strategy by first considering the dynamics of the covariance matrix $\bm{\Sigma}(t)$.

Let us define an auxiliary matrix $\bm{\Sigma}(t)$ as the solution to the following differential Lyapunov equation, with the initial condition $\bm{\Sigma}(0) = \bm{G}^{-1}(0)$:
\begin{equation}
    \frac{\mathrm{d}\bm{\Sigma}}{\mathrm{d}t} = - \bm{\Sigma}(\Gamma_2+\mathcal{M}-\bm{C}^T) - (\Gamma_2+\mathcal{M}-\bm{C})\bm{\Sigma} - (\Gamma_1-\bm{B}).
    \label{eq:Sigma_dynamics_open_proof}
\end{equation}
Since the initial matrix $\bm{G}(0)$ is positive definite, its inverse $\bm{\Sigma}(0)$ is also positive definite. According to Proposition \ref{prop:positive_definiteness_of_X}, the solution $\bm{\Sigma}(t)$ to the Lyapunov equation above remains positive definite for all $t > 0$. This implies that its inverse, which we denote as $\widetilde{\bm{G}}(t) = \bm{\Sigma}^{-1}(t)$, is also well-defined and positive definite.

We now show that this matrix $\widetilde{\bm{G}}(t)$ satisfies the same initial value problem as our original matrix $\bm{G}(t)$. By differentiating $\widetilde{\bm{G}}(t) = \bm{\Sigma}^{-1}(t)$ and substituting \cref{eq:Sigma_dynamics_open_proof}, we find:
\begin{equation}
\begin{aligned}
    \frac{\mathrm{d}\widetilde{\bm{G}}}{\mathrm{d}t} &= -\bm{\Sigma}^{-1} \left(\frac{\mathrm{d}\bm{\Sigma}}{\mathrm{d}t}\right) \bm{\Sigma}^{-1} \\
    &= \bm{\Sigma}^{-1} \left( \bm{\Sigma}(\Gamma_2+\mathcal{M}-\bm{C}^T) + (\Gamma_2+\mathcal{M}-\bm{C})\bm{\Sigma} + (\Gamma_1-\bm{B}) \right) \bm{\Sigma}^{-1} \\
    &= (\Gamma_2+\mathcal{M}-\bm{C}^T)\widetilde{\bm{G}} + \widetilde{\bm{G}}(\Gamma_2+\mathcal{M}-\bm{C}) + \widetilde{\bm{G}}(\Gamma_1-\bm{B})\widetilde{\bm{G}}.
\end{aligned}
\end{equation}
This is identical to the equation of motion for $\bm{G}(t)$, \cref{eq:G_dynamics_open}. Since $\widetilde{\bm{G}}(0) = \bm{\Sigma}^{-1}(0) = \bm{G}(0)$, by the uniqueness of solutions to ODEs, we must have $\widetilde{\bm{G}}(t) = \bm{G}(t)$ for all $t \ge 0$.

Therefore, since $\bm{G}(t)$ is identical to the positive definite matrix $\widetilde{\bm{G}}(t)$, it must also be positive definite. This completes the proof.
\end{proof}

\begin{remark}
While the derived equations of motion \eqref{eq:G_dynamics_open}-\eqref{eq:p_dynamics_open} are equivalent to results for Gaussian wavepacket dynamics in open systems previously reported in the literature \cite{graefe2018lindblad,hernandez2025classical}, our derivation is based on rigorous asymptotic analysis. The advantage of the formulation presented here lies in its clear and compact matrix representation. Specifically, our approach distinctly separates the conservative Hamiltonian evolution, governed by the matrix $\bm{C}(t)$, from the non-Hamiltonian effects, governed by the matrices $\bm{B}$ and $\Gamma_{1,2}$. This separation not only enhances the conceptual clarity of the dynamics but also makes the system of equations particularly well-suited for robust numerical integration.
\end{remark}

\section{The Gaussian Sampling Algorithm}
\label{sec:algorithm}
In this section, we present the full construction of our Frozen Gaussian sampling (FGS) algorithm. The algorithm is based on a two-step procedure: (1) an initial decomposition of the Wigner function into a weighted superposition of narrower Gaussian wavepackets, and (2) a Monte Carlo sampling approach where the evolution of each packet is tracked according to the dynamics derived in Section \ref{sec:wavepacket_dynamics}.

\subsection{Initial Value Decomposition}
The FGS algorithm is, in principle, applicable to general initial conditions. The first step, however, requires decomposing the initial state into a superposition of $\varepsilon$-scaled Gaussian wavepackets. While developing such a decomposition for arbitrary functions is a subject for future research, in this work, we focus on the important and illustrative case of a general Gaussian initial value:
\begin{equation}
    W(\boldsymbol{x},\boldsymbol{\xi},0)=\frac{1}{(2\pi)^n\sqrt{\det(\boldsymbol{\Sigma}_0)}}\exp\left(-\frac{1}{2}\begin{pmatrix}
        \boldsymbol{x}-\boldsymbol{x}_0\\
        \boldsymbol{\xi}-\boldsymbol{\xi}_0
    \end{pmatrix}^T\boldsymbol{\Sigma}_0^{-1}\begin{pmatrix}
        \boldsymbol{x}-\boldsymbol{x}_0\\
        \boldsymbol{\xi}-\boldsymbol{\xi}_0
    \end{pmatrix}\right).
    \label{eq: Gaussian initial value}
\end{equation}

For this initial state, the required decomposition can be achieved through the following exact identity, which can be verified by direct calculation:
\begin{equation}
     W(\boldsymbol{x},\boldsymbol{\xi},0)= \int W_{\varepsilon}(\boldsymbol{x},\boldsymbol{\xi};\boldsymbol{q},\boldsymbol{p})\pi(\boldsymbol{q},\boldsymbol{p})\mathrm{d}\boldsymbol{q}\mathrm{d}\boldsymbol{p},
\end{equation}
where the $W_{\varepsilon}$ is an $\varepsilon$-scaled Gaussian wavepacket centered at $(\bm{q},\bm{p})$, and $\pi(\boldsymbol{q},\boldsymbol{p})$ is the sampling distribution for the wavepacket centers:
\begin{align}
     &W_{\varepsilon}(\boldsymbol{x},\boldsymbol{\xi};\boldsymbol{q},\boldsymbol{p}) = \frac{1}{(2\pi\varepsilon)^n\sqrt{\det(\boldsymbol{\Sigma}_0)}}\exp\left(-\frac{1}{2\varepsilon}\begin{pmatrix}
         \boldsymbol{x}-\boldsymbol{q}\\
         \boldsymbol{\xi}-\boldsymbol{p}
     \end{pmatrix}^T\boldsymbol{\Sigma}_0^{-1}\begin{pmatrix}
         \boldsymbol{x}-\boldsymbol{q}\\
         \boldsymbol{\xi}-\boldsymbol{p}
     \end{pmatrix}\right), \label{eq:epsilon_scaled_kernel} \\
     &\pi(\boldsymbol{q},\boldsymbol{p}) = \frac{1}{(2\pi(1-\varepsilon))^n\sqrt{\det(\boldsymbol{\Sigma}_0)}}\exp\left(-\frac{1}{2(1-\varepsilon)}\begin{pmatrix}
         \boldsymbol{x}_0-\boldsymbol{q}\\
         \boldsymbol{\xi}_0-\boldsymbol{p}
     \end{pmatrix}^T\boldsymbol{\Sigma}_0^{-1}\begin{pmatrix}
         \boldsymbol{x}_0-\boldsymbol{q}\\
         \boldsymbol{\xi}_0-\boldsymbol{p}
     \end{pmatrix}\right). \label{eq:initial_Gaussian_distribution}
\end{align}
The function $\pi(\bm{q},\bm{p})$ is a Gaussian probability density function (PDF) in phase space, centered at $(\bm{x}_0, \bm{\xi}_0)$ with covariance $(1-\varepsilon)\bm{\Sigma}_0$. This probabilistic interpretation is fundamental to our sampling algorithm, as it provides the distribution from which the initial wavepacket centers are drawn.

The linearity of the WFP equation allows us to obtain the solution at time $t$ by evolving each component of the integral superposition independently:
\begin{equation}
    W(\bm{x},\bm{\xi},t) = \int W(\bm{x},\bm{\xi};\bm{q},\bm{p},t) \, \pi(\bm{q},\bm{p}) \, \mathrm{d}\bm{q}\mathrm{d}\bm{p},
\end{equation}
where $W(\bm{x},\bm{\xi};\bm{q},\bm{p},t)$ is the exact solution to the WFP equation with the initial condition given by the  $W_{\varepsilon}(\bm{x},\bm{\xi};\bm{q},\bm{p})$.

Our semiclassical approximation, denoted $W_{\mathrm{FGS}}$, consists of replacing the exact evolution of the kernel with a single evolving Gaussian wavepacket:
\begin{equation}
    W(\bm{x},\bm{\xi};\bm{q},\bm{p},t) \approx W_{\mathrm{FGS}}(\bm{x},\bm{\xi};\bm{q},\bm{p},t) = A(t) \exp\left(-\frac{1}{2\varepsilon}T(\bm{x}, \bm{\xi}, \bm{q}(t), \bm{p}(t); \bm{\Sigma}(t))\right).
    \label{eq:wigner_function_for_single_sample}
\end{equation}
The parameters of this Gaussian are propagated according to the ODEs \eqref{eq:G_dynamics_open}-\eqref{eq:p_dynamics_open} with initial conditions determined by the starting kernel $W_{\varepsilon}$:
\begin{equation}
    \bm{q}(0)=\bm{q}, \quad \bm{p}(0)=\bm{p}, \quad \bm{\Sigma}(0)=\bm{\Sigma}_0, \quad A(0) = \frac{1}{(2\pi\varepsilon)^n\sqrt{\det(\bm{\Sigma}_0)}}.
    \label{eq:single_packet_initial_conditions}
\end{equation}
The full FGS approximation for the Wigner function at time $t$ is then constructed by integrating these evolved kernels against the initial sampling distribution:
\begin{equation}
    W(\bm{x},\bm{\xi},t) \approx W_{\mathrm{FGS}}(\bm{x},\bm{\xi},t) = \int W_{\mathrm{FGS}}(\bm{x},\bm{\xi};\bm{q},\bm{p},t) \, \pi(\bm{q},\bm{p}) \, \mathrm{d}\bm{q}\mathrm{d}\boldsymbol{p}.
    \label{eq:final_approximation_integral}
\end{equation}

\subsection{Construction of the algorithm}
The integral representation of the solution, \cref{eq:final_approximation_integral}, is the cornerstone of our sampling algorithm. By interpreting $\pi(\bm{q},\bm{p})$ as a PDF, the integral can be written as an expectation over the random vector $\bm{z}=(\bm{q}^T, \bm{p}^T)^T \sim \pi(\bm{z})$:
\begin{equation}
    W_{\mathrm{FGS}}(\bm{x},\bm{\xi},t) = \mathbb{E}_{\bm{z}\sim \pi}[W_{\mathrm{FGS}}(\bm{x},\bm{\xi};\bm{z},t)],
    \label{eq:expectation_form}
\end{equation}
where $W_{\mathrm{FGS}}(\bm{x},\bm{\xi};\bm{z},t)$ denotes the evolved wavepacket with initial center $\bm{z}$.

This expectation form is naturally suited for a Monte Carlo approximation. By the law of large numbers, the integral can be approximated by the sample mean of $M$ independent and identically distributed (i.i.d.) samples $\{\bm{z}_j\}_{j=1}^M$ drawn from the distribution $\pi(\bm{z})$:
\begin{equation}
    W_{\mathrm{FGS}}(\bm{x},\bm{\xi},t) \approx \frac{1}{M}\sum_{j=1}^M W_{\mathrm{FGS}}(\bm{x},\bm{\xi};\bm{z}_j,t).
    \label{eq:monte_carlo_approximation}
\end{equation}
The Monte Carlo approximation \eqref{eq:monte_carlo_approximation} naturally leads to a practical sampling algorithm, which consists of three main stages.

First, in the initialization step, we generate $M$ independent and identically distributed (i.i.d.) samples for the wavepacket centers, $\{\bm{z}_j\}_{j=1}^M$, by drawing from the Gaussian PDF $\pi(\bm{z})$ defined in \cref{eq:initial_Gaussian_distribution}.

Second, in the evolution step, each sample $\bm{z}_j$ is used to set the initial conditions for a corresponding wavepacket. Specifically, the sample defines the initial center $(\bm{q}_j(0), \bm{p}_j(0)) = \bm{z}_j$, while the initial amplitude $A_j(0)$ and covariance matrix $\bm{\Sigma}_j(0)$ are identical for all packets, given by \cref{eq:single_packet_initial_conditions}. The full set of parameters for each of the $M$ wavepackets is then propagated from $t=0$ to the final time $t=T$ by numerically integrating the system of ODEs \eqref{eq:G_dynamics_open}-\eqref{eq:p_dynamics_open}.

Finally, in the reconstruction step, the full Wigner function at the final time, $W_{\mathrm{FGS}}(\bm{x},\bm{\xi},T)$, is constructed by taking the sample mean of all $M$ evolved wavepackets, as prescribed by \cref{eq:monte_carlo_approximation}.

The complete procedure is formally summarized in Algorithm \ref{alg:sampling}.

\begin{algorithm}
\caption{The Frozen Gaussian Sampling (FGS) Algorithm}
\label{alg:sampling}
\begin{algorithmic}[1] 
\STATE \textbf{Input:} Initial state parameters $\bm{x}_0, \bm{\xi}_0, \bm{\Sigma}_0$; number of samples $M$; final time $T$.
\STATE \textbf{Output:} The approximated Wigner function $W_{\mathrm{FGS}}(\bm{x},\bm{\xi},T)$.

\STATE \textbf{Step 1: Initialization}
\STATE Generate $M$ i.i.d. samples $\{\bm{z}_j = (\bm{q}_j^T, \bm{p}_j^T)^T\}_{j=1}^M$ from the Gaussian PDF $\pi(\bm{z})$ in \cref{eq:initial_Gaussian_distribution}.

\STATE \textbf{Step 2: Evolution}
\FOR{$j = 1$ to $M$}
    \STATE Set initial conditions for the $j$-th wavepacket: 
    $$ \bm{q}_j(0)=\bm{q}_j, \quad \bm{p}_j(0)=\bm{p}_j, \quad \bm{\Sigma}_j(0)=\bm{\Sigma}_0, \quad A_j(0) = \frac{1}{(2\pi\varepsilon)^n\sqrt{\det(\bm{\Sigma}_0)}}. $$
    \STATE Integrate the ODEs \eqref{eq:G_dynamics_open}-\eqref{eq:p_dynamics_open} for the wavepacket parameters from $t=0$ to $t=T$ using a numerical solver (e.g., fourth-order Runge-Kutta) to obtain the final parameters $\{\bm{q}_j(T), \bm{p}_j(T), A_j(T), \bm{\Sigma}_j(T)\}$.
\ENDFOR

\STATE \textbf{Step 3: Reconstruction}
\STATE Reconstruct the final Wigner function using the sample mean: $$W_{\mathrm{FGS}}(\bm{x},\bm{\xi},T) = \frac{1}{M}\sum_{j=1}^M A_j(T) \exp\left(-\frac{1}{2\varepsilon}T(\bm{x}, \bm{\xi}, \bm{q}_j(T), \bm{p}_j(T); \bm{\Sigma}_j(T))\right).$$
\STATE \textbf{return} $W_{\mathrm{FGS}}(\bm{x},\bm{\xi},T)$.
\end{algorithmic}
\end{algorithm}

\subsection{Error Analysis for Physical Observables}

In this subsection, we analyze the statistical sampling error inherent in the FGS algorithm when computing physical observables. A central objective of this analysis is to quantify the error's sensitivity to the semiclassical parameter $\varepsilon$.

Such sensitivity poses a well-known computational challenge, as the cost of many semiclassical methods scales unfavorably with $\varepsilon$, with the notable exception of the FGS applied to the Schr\"odinger equation with Gaussian initial data \cite{xie2024frozen, huang2023efficient, chai2024frozen}. We observe a similar dependency in the general application of our method; specifically, when the FGS is used to directly compute the Wigner function, the sampling error exhibits scaling with the semiclassical parameter.

However, a distinct advantage emerges when the FGS algorithm is employed to compute physical observables. In this context, the sampling error becomes independent of $\varepsilon$. We now provide a rigorous proof of this property, which constitutes a primary merit of our approach.

Consider the statistical sampling error for the expectation value of a general physical observable, $a(\bm{x},\bm{\xi})$. 
Let $\langle a(t) \rangle_{\mathrm{FGS}}$ denote the FGS approximation of the expectation value, computed using the full integral solution from \cref{eq:final_approximation_integral}:
\begin{equation}
     \langle a(t) \rangle_{\mathrm{FGS}} := \iint a(\bm{x},\bm{\xi})W_{\mathrm{FGS}}(\bm{x},\bm{\xi},t)\,\mathrm{d}\bm{x}\mathrm{d}\bm{\xi}.
\end{equation}
By substituting the integral form of $W_{\mathrm{FGS}}$ and swapping the order of integration, this becomes an expectation over the sampling distribution $\pi(\bm{z})$:
\begin{equation}
    \langle a(t) \rangle_{\mathrm{FGS}} = \int a(\bm{z},t) \pi(\bm{z}) \,\mathrm{d}\bm{z} = \mathbb{E}_{\bm{z}\sim\pi}[a(\bm{z},t)],
    \label{eq:fgs_integral_expectation}
\end{equation}
where $a(\bm{z},t)$ is the observable of a single evolved wavepacket, $W_{\mathrm{FGS}}(\bm{x},\bm{\xi};\bm{z},t)$:
\begin{equation}
     a(\bm{z},t) = \iint a(\bm{x},\bm{\xi}) W_{\mathrm{FGS}}(\bm{x},\bm{\xi};\bm{z},t)\,\mathrm{d}\bm{x}\mathrm{d}\bm{\xi}.
\end{equation}

The FGS algorithm approximates this integral value $\langle a(t) \rangle_{\mathrm{FGS}}$ with the $M$-sample mean, $\langle a \rangle_{\mathrm{FGS}}^M = \frac{1}{M}\sum_{j=1}^M a(\bm{z}_j,t)$. The statistical sampling error $\mathcal{E}_{a}$ is the root mean squared error (RMSE) of this Monte Carlo approximation:
\begin{equation}
    \mathcal{E}_{a}(\varepsilon,M,t) := \left( \mathbb{E}_{\left\{\bm{z}_j\right\}_{j=1}^M\sim\pi}\left[ \left( \langle a(t) \rangle_{\mathrm{FGS}} - \frac{1}{M}\sum_{j=1}^M a(\bm{z}_j,t) \right)^2 \right] \right)^{1/2} = \sqrt{\frac{\mathrm{Var}_{\bm{z}\sim\pi}[a(\bm{z},t)]}{M}},
    \label{eq:rmse_error_definition}
\end{equation}
where $\mathrm{Var}[X] = \mathbb{E}[X^2] - (\mathbb{E}[X])^2$ is the variance.

\begin{theorem}[Error Estimate for Lipschitz and Quadratic Observables]
\label{thm:error_estimate}
We assume the potential $V(x)$ satisfies the subquadratic condition, i.e., its second-order derivative is globally bounded:
\begin{equation}
    \sup_{q \in \mathbb{R}} |V''(q)| \le C_V < \infty.
    \label{eq:subquadratic_condition}
\end{equation}
For a one-dimensional system with a Gaussian initial condition \eqref{eq: Gaussian initial value}, the root-mean-square sampling error of the FGS algorithm for:
\begin{itemize}
    \item[(i)] \textbf{Lipschitz continuous observables} $a(x,\xi)$ (with Lipschitz constant $L_a$), and
    \item[(ii)] \textbf{Quadratic observables} $a(x,\xi)$ (e.g., kinetic energy, harmonic potential),
\end{itemize}
is bounded by:
\begin{equation}
    \mathcal{E}_{a}(\varepsilon,M,t) \le C(t) \sqrt{\frac{1-\varepsilon}{M}},
    \label{eq:error_estimate}
\end{equation}
where the constant $C(t)$ depends on time but is independent of $M$ and $\varepsilon$.
\end{theorem}

\begin{proof}
\textbf{Case 1: Linear Observables.}
For brevity, we present the proof for $a=x$. The squared RMSE \eqref{eq:rmse_error_definition} is the variance of the evolved center position, scaled by $M$:
\begin{equation}
    (\mathcal{E}_{x})^2 = \frac{1}{M}\mathrm{Var}[q(t)] = \frac{1}{M}\mathbb{E}_{\bm{z}\sim\pi}\left[ (q(t) - \langle q(t) \rangle_{\mathrm{FGS}})^2 \right].
    \label{eq:var_def}
\end{equation}
The core of the proof is to bound this variance. Let $\bar{q}(t)$ be the "mean trajectory," i.e., the trajectory initiated at the mean of the initial distribution, $\bar{z}(0) = \mathbb{E}[\bm{z}(0)] = (x_0, \xi_0)$.

A fundamental property of variance is that it is the minimum second moment. That is, for any constant $c$, $\mathrm{Var}[X] = \mathbb{E}[(X - \mathbb{E}[X])^2] \le \mathbb{E}[(X-c)^2]$. We can therefore bound the true variance by the second moment around the (easier to analyze) mean trajectory $\bar{q}(t)$:
\begin{equation}
    \mathrm{Var}[q(t)] \le \mathbb{E}_{\bm{z}\sim\pi}\left[ (q(t) - \bar{q}(t))^2 \right].
    \label{eq:var_bound_by_mean_traj}
\end{equation}
Our task now is to bound the term on the right-hand side. Let $\delta q(t) = q(t) - \bar{q}(t)$ be the deviation from the mean trajectory. This deviation is a function of the initial deviations $\delta q(0) = q(0) - x_0$ and $\delta p(0) = p(0) - \xi_0$.

To do this, we analyze the sensitivity of the trajectory to its initial conditions. Let $\bm{\chi}_{q_0}(t) = (\partial_{q(0)}q(t), \partial_{q(0)}p(t))^T$. Differentiating the equations of motion \eqref{eq:q_dynamics_open} and \eqref{eq:p_dynamics_open} with respect to $q(0)$ and applying the chain rule yields:
\begin{align}
    \partial_{q(0)}\dot{q}(t) &= \partial_{q(0)}p(t) + (\operatorname{Im}(\gamma) + \mu)\partial_{q(0)}q(t), \\
    \partial_{q(0)}\dot{p}(t) &= -\left(\nabla_x^2 V|_{x=q(t)}\right)\partial_{q(0)}q(t) + (\operatorname{Im}(\gamma) - \mu)\partial_{q(0)}p(t).
\end{align}
This is a linear ODE system $\dot{\bm{\chi}}_{q_0}(t) = A(t)\bm{\chi}_{q_0}(t)$, where the matrix $A(t)$ is
\begin{equation}
    A(t) = \begin{pmatrix}
    \operatorname{Im}(\gamma) + \mu & 1 \\
    -V''(q(t)) & \operatorname{Im}(\gamma) - \mu
    \end{pmatrix}.
\end{equation}
Crucially, because we assume the subquadratic condition \eqref{eq:subquadratic_condition}, the term $|V''(q(t))|$ is globally bounded by $C_V$. This ensures that the matrix $A(t)$ is uniformly bounded for all $t$ and all trajectories $q(t)$. By Gr\"onwall's inequality, the solution $\bm{\chi}_{q_0}(t)$ is bounded, meaning there exists a constant $C_1(t)$ such that $|\partial_{q(0)}q(t)| \le C_1(t)$.

A parallel argument for derivatives with respect to $p(0)$ shows that $|\partial_{p(0)}q(t)|$ is also bounded by some constant $C_2(t)$.

By the Mean Value Theorem, the deviation $\delta q(t)$ is:
\begin{equation*}
    \delta q(t) = \left. \frac{\partial q}{\partial q(0)} \right|_{\tilde{z}(0)} \delta q(0) + \left. \frac{\partial q}{\partial p(0)} \right|_{\tilde{z}(0)} \delta p(0),
\end{equation*}
where $\tilde{z}(0)$ is a point between $z(0)$ and $\bar{z}(0)$. We can now use the bounds on the derivatives and the inequality $(a+b)^2 \le 2a^2 + 2b^2$ to bound the square of the deviation:
\begin{align*}
    (\delta q(t))^2 &\le 2 \left( \left|\frac{\partial q}{\partial q(0)}\right|^2 |\delta q(0)|^2 + \left|\frac{\partial q}{\partial p(0)}\right|^2 |\delta p(0)|^2 \right) \\
    &\le 2 \left( C_1(t)^2 (\delta q(0))^2 + C_2(t)^2 (\delta p(0))^2 \right).
\end{align*}
Taking the expectation of both sides (which is the term we wanted to bound in \cref{eq:var_bound_by_mean_traj}):
\begin{align*}
    \mathbb{E}[(\delta q(t))^2] 
    &\le 2 C_1(t)^2 \mathbb{E}[(\delta q(0))^2] + 2 C_2(t)^2 \mathbb{E}[(\delta p(0))^2] \\
    &= 2 C_1(t)^2 \mathrm{Var}[q(0)] + 2 C_2(t)^2 \mathrm{Var}[p(0)].
\end{align*}
(The last step holds because $\mathbb{E}[\delta q(0)] = \mathbb{E}[q(0) - x_0] = 0$).

The initial state $\bm{z}(0)$ is sampled from the PDF $\pi(\bm{z})$ \eqref{eq:initial_Gaussian_distribution}, which has a covariance matrix $\bm{\Sigma}_\pi = (1-\varepsilon)\bm{\Sigma}_0$. Therefore, both initial variances are directly proportional to $(1-\varepsilon)$:
\begin{equation}
    \mathrm{Var}[q(0)] \propto (1-\varepsilon) \quad \text{and} \quad \mathrm{Var}[p(0)] \propto (1-\varepsilon).
\end{equation}
Since $\mathrm{Var}[q(t)] \le \mathbb{E}[(\delta q(t))^2]$, and $\mathbb{E}[(\delta q(t))^2]$ is bounded by a sum of terms all proportional to $(1-\varepsilon)$, we conclude that:
\begin{equation}
    \mathrm{Var}[q(t)] \le C'(t) (1-\varepsilon).
    \label{eq:linear_result_summary}
\end{equation}
Substituting this rigorous bound back into \cref{eq:var_def} yields the final estimate for linear observables: $\mathcal{E}_{x} \le C(t) \sqrt{(1-\varepsilon)/M}$.

\textbf{Case 2: Lipschitz Observables.}
For a general observable $a(z)$ that is Lipschitz continuous in phase space with constant $L_a$, the result follows immediately from the linear bound established above. Since $|a(z(t)) - a(\bar{z}(t))| \le L_a \|z(t) - \bar{z}(t)\|$, squaring and taking the expectation yields $\mathbb{E}[(a - \bar{a})^2] \le L_a^2 \mathbb{E}[\|\delta z(t)\|^2]$. Using the bound from \eqref{eq:linear_result_summary}, the variance is proportional to $(1-\varepsilon)$, preserving the convergence rate.

\textbf{Case 3: Quadratic Observables.}
Consider a general quadratic observable $a(z) = z^T \mathbf{Q} z + \mathbf{b}^T z + c$. We perform a Taylor expansion of $a(z(t))$ around the mean trajectory $\bar{z}(t)$. Since $a(z)$ is quadratic, the expansion terminates at the second order and is exact:
\begin{equation}
    a(z(t)) - a(\bar{z}(t)) = \nabla a(\bar{z}(t))^T \delta z(t) + \delta z(t)^T \mathbf{Q} \delta z(t),
\end{equation}
where $\delta z(t) = z(t) - \bar{z}(t)$. Using the inequality $(u+v)^2 \le 2u^2 + 2v^2$, the second moment is bounded by:
\begin{equation}
    \mathbb{E}[(a(z) - a(\bar{z}))^2] \le 2 \underbrace{\mathbb{E}[|\nabla a^T \delta z|^2]}_{I_1} + 2 \underbrace{\mathbb{E}[|\delta z^T \mathbf{Q} \delta z|^2]}_{I_2}.
\end{equation}

For the \textbf{Linear part ($I_1$)}, let $K(t) = \|\nabla a(\bar{z}(t))\|$. Using the result from Case 1 \eqref{eq:linear_result_summary}, we have:
\begin{equation}
    I_1 \le K(t)^2 \mathbb{E}[\|\delta z(t)\|^2] \le K(t)^2 C'(t) (1-\varepsilon).
\end{equation}

For the \textbf{Quadratic part ($I_2$)}, we need to bound the fourth moment. Recall from Case 1 that the deviation satisfies the pointwise bound $\|\delta z(t)\| \le C(t) \|\delta z(0)\|$ (derived via the bounded Jacobian matrix). Therefore:
\begin{equation}
    |\delta z(t)^T \mathbf{Q} \delta z(t)| \le \|\mathbf{Q}\| \|\delta z(t)\|^2 \le \|\mathbf{Q}\| C(t)^2 \|\delta z(0)\|^2.
\end{equation}
Squaring this and taking the expectation:
\begin{equation}
    I_2 = \mathbb{E}[|\delta z(t)^T \mathbf{Q} \delta z(t)|^2] \le \|\mathbf{Q}\|^2 C(t)^4 \mathbb{E}[\|\delta z(0)\|^4].
\end{equation}
The initial deviation $\delta z(0)$ follows a Gaussian distribution with covariance $\bm{\Sigma}_\pi = (1-\varepsilon)\bm{\Sigma}_0$. For a Gaussian vector $X \sim \mathcal{N}(0, \Sigma)$, the fourth moment scales as the square of the variance, specifically $\mathbb{E}[\|X\|^4] \le C_d (\operatorname{tr}(\Sigma))^2$. Thus:
\begin{equation}
    \mathbb{E}[\|\delta z(0)\|^4] = C_0 (1-\varepsilon)^2.
\end{equation}
Substituting this back, we obtain the bound for the quadratic term:
\begin{equation}
    I_2 \le \|\mathbf{Q}\|^2 C(t)^4 C_0 (1-\varepsilon)^2 = C_{quad}(t) (1-\varepsilon)^2.
\end{equation}

Combining the bounds for $I_1$ and $I_2$, the total variance is bounded by:
\begin{equation}
    \mathrm{Var}[a(z(t))] \le 2 K(t)^2 C'(t) (1-\varepsilon) + 2 C_{quad}(t) (1-\varepsilon)^2.
\end{equation}
Since $0 < \varepsilon < 1$, the linear term $(1-\varepsilon)$ dominates the quadratic term $(1-\varepsilon)^2$. Thus, the error scales as:
\begin{equation}
    \mathcal{E}_{a} = \sqrt{\frac{\mathrm{Var}}{M}} \le C(t) \sqrt{\frac{1-\varepsilon}{M}}.
\end{equation}

This completes the proof.
\end{proof}

\begin{corollary}[Error Estimate for the Wigner Equation]
\label{cor:wigner_error}
In the absence of dissipation (i.e., for the closed-system Wigner equation \eqref{eq:Wigner_equation}) with a Gaussian initial condition, the same error estimate \eqref{eq:error_estimate} holds for the computation of linear observables.
\end{corollary}

To appreciate the significance of this result, it is instructive to compare it with the time-splitting spectral method, which is a leading approach for computing observables in closed quantum systems. It has been shown that this method can compute observables accurately with a time step independent of $\varepsilon$. However, its fundamental limitation as a grid-based method is the requirement to resolve the $O(\varepsilon)$ spatial oscillations of the quantum state, meaning the computational cost associated with spatial resolution remains high \cite{bao2002time,golse2021convergence}. In contrast, our FGS algorithm, by its very nature as a sampling method, is inherently mesh-free. As highlighted in the introduction, our approach replaces the challenge of resolving a single, complex function on a grid with the task of evolving $M$ independent wavepackets. The result in Corollary \ref{cor:wigner_error} therefore implies that our method, which only pays a statistical cost related to $M$, is robust with respect to $\varepsilon$ in both time and space. This highlights a key advantage of our sampling approach for efficiently computing physical observables in the semiclassical regime.

\section{Experimental results}
\label{sec:experiments}
In this section, we present a series of numerical experiments to demonstrate the performance of our proposed Frozen Gaussian sampling (FGS) algorithm. The section is structured as follows. First, we validate the accuracy and convergence properties of the FGS algorithm against benchmark solutions. Second, we demonstrate its superior stability for long-time simulations compared to traditional grid-based methods. Finally, we apply the validated algorithm to conduct numerical explorations of the existence of steady states for various non-harmonic potentials.

Unless otherwise specified, all experiments use a Gaussian wave packet as the initial condition:
\begin{equation} 
\label{eq:initial_condition}
W(x,\xi,0)=\frac{1}{2\pi\sqrt{\det(\Sigma_0)}}\exp{\left(-\frac{1}{2}\begin{pmatrix}
	x-x_0\\
	\xi-\xi_0
\end{pmatrix}^T\Sigma_0^{-1}\begin{pmatrix}
	x-x_0\\
	\xi-\xi_0
\end{pmatrix}\right)},
\end{equation}
where $\Sigma_0$ is the initial covariance matrix. This initial condition represents a wave packet centered at $(x_0, \xi_0)$ in phase space. The system of ordinary differential equations for the wavepacket parameters is integrated using the fourth-order Runge-Kutta method with a time step of $\Delta t=0.01$.


\subsection{Validation of Accuracy and Convergence}
\label{subsec:validation}
In this subsection, we validate the accuracy and convergence of the FGS algorithm. For potentials where a fully analytical solution is unavailable, we use highly-resolved numerical or semi-analytical solutions as benchmarks. We will demonstrate the expected convergence rates of the FGS algorithm for both the Wigner function and key physical observables, with respect to the sample size $M$ and the semiclassical parameter $\varepsilon$. Furthermore, we will show its superior performance in long-time simulations where the TSSP method can be limited by its fixed computational domain.

For the convergence tests, the reference TSSP solution is computed on a phase space domain of $[-4,4]\times[-4,4]$ with a fine grid resolution of $\mathrm{d}x=\mathrm{d}\xi=0.002$. The simulation runs until a final time of $T=0.5$ with a TSSP time step of $\mathrm{d}t=0.001$. Unless stated otherwise, all numerical experiments are performed for a one-dimensional system ($n=1$). Consequently, we drop the component index $k$ and set the physical parameters in \cref{eq:dissipative_coeffs} to $\alpha=0.4, \beta=0.1,\gamma=1\mathrm{i}$ and $\mu=-1$. The initial covariance matrix is $\Sigma_0 = \mathrm{diag}(5, 5)$.

To quantify the accuracy, we evaluate two types of errors. The error of the Wigner function itself is measured by the relative $L^2$ norm:
\begin{equation}
\label{eq:l2_error_wigner}
\mathrm{error}_{L^2}(W_{\mathrm{FGS}}) := \frac{\|W_{\mathrm{FGS}}-W_{\mathrm{ref}}\|_{L^2}}{\|W_{\mathrm{ref}}\|_{L^2}} = \frac{\left(\int |W_{\mathrm{FGS}}-W_{\mathrm{ref}}|^2\mathrm{d}x\mathrm{d}\xi\right)^{1/2}}{\left(\int |W_{\mathrm{ref}}|^2\mathrm{d}x\mathrm{d}\xi\right)^{1/2}}.
\end{equation}
For a physical observable $a(x,\xi)$, we compute the absolute error of its expectation value:
\begin{equation}
\label{eq:abs_error_obs}
\mathrm{error}(\langle a \rangle) := |\langle a \rangle_{\mathrm{FGS}} - \langle a \rangle_{\mathrm{ref}}|.
\end{equation}
The reference solution $W_{\mathrm{ref}}$ and its corresponding observables are obtained as follows:
\begin{itemize}
    \item For the harmonic potential, the Wigner function is known to remain a Gaussian for all time. We obtain a high-precision semi-analytical reference solution by numerically solving the system of ODEs that governs the evolution of the Gaussian's parameters (i.e., the mean and covariance matrix).
    \item For other potentials, the reference solution is the benchmark result obtained from the TSSP \cite{yi2025time} method using the aforementioned numerical parameters.
\end{itemize}
\begin{remark}
    Since $W_{\mathrm{FGS}}$ and $\langle a \rangle_{\mathrm{FGS}}$ are Monte Carlo estimates, the errors defined in \cref{eq:l2_error_wigner} and \cref{eq:abs_error_obs} are random variables. We therefore compute their expected values, which are estimated in our tables by the sample mean over $N_{\text{repeat}}=500$ independent runs.
\end{remark}

\subsubsection{Example 1: Harmonic Potential.}
We begin with a harmonic potential, $V(x)=\frac{1}{2}x^2+x$, for which a semi-analytical solution is available. The initial wave packet is centered at $(x_0, \xi_0) = (-0.1, 0.2)$. This example serves to numerically validate our theoretical error analysis and to demonstrate the most significant advantage of the FGS algorithm: its robustness to the semiclassical parameter $\varepsilon$ when computing physical observables.

The numerical errors for the observables $\langle x \rangle$ and $\langle \xi \rangle$ are presented in Tables \ref{table:example 1-x} and \ref{table:example 1-xi}, respectively. The data clearly demonstrates the primary highlight of our algorithm: for a fixed sample size $M$, the errors are remarkably independent of the semiclassical parameter $\varepsilon$. For instance, in Table \ref{table:example 1-x}, as $\varepsilon$ decreases from $1/16$ to $1/128$, the error for $M=800$ barely changes (from $0.011990$ to $0.012335$). This result is a direct numerical confirmation of our theoretical finding in Theorem \ref{thm:error_estimate}, which showed the error scales with $\sqrt{1-\varepsilon}$.

Second, for a fixed $\varepsilon$, the errors exhibit the expected $\mathcal{O}(M^{-1/2})$ convergence rate, which is consistent with the statistical nature of a Monte Carlo method. These two trends are visualized in the log-log plots in Figure \ref{fig:example 1-log plot}. The tight clustering of the error curves for different $\varepsilon$ values visually confirms the algorithm's robustness to the semiclassical parameter, while the near-parallel slopes of the lines confirm the $\mathcal{O}(M^{-1/2})$ scaling.

Finally, Figure \ref{fig:example 1} provides a qualitative comparison, showing an excellent visual agreement between the Wigner function computed by the FGS algorithm (with $M=3200$) and the reference solution at $T=1.0$. This demonstrates that our method accurately captures the key phase space dynamics, including the rotation and shearing of the wavepacket.

\begin{table}[htbp]
\centering
\caption{Absolute errors of the mean position $\langle x\rangle$ for Example 1 at $T=1.0$. The table shows the errors as a function of the sample size $M$ and the semiclassical parameter $\varepsilon$. The reference solution is the semi-analytical one described in Section \ref{subsec:validation}.}
\label{table:example 1-x}
\begin{tabular}{c|cccc}
\hline
Error & $\varepsilon = 1/16$ & $\varepsilon = 1/32$ & $\varepsilon = 1/64$ & $\varepsilon = 1/128$ \\
\hline
$M = 100$ & 0.032456 & 0.032992 & 0.033257 & 0.033389 \\
$M = 200$ & 0.023355 & 0.023741 & 0.023931 & 0.024026 \\
$M = 400$ & 0.016967 & 0.017248 & 0.017386 & 0.017455 \\
$M = 800$ & 0.011990 & 0.012189 & 0.012287 & 0.012335 \\
$M = 1600$ & 0.0084869 & 0.0086272 & 0.0086965 & 0.0087309 \\
\hline
\end{tabular}
\end{table}

\begin{table}[htbp]
\centering
\caption{Absolute errors of the mean momentum $\langle \xi\rangle$ for Example 1 at $T=1.0$. The table shows the errors as a function of the sample size $M$ and the semiclassical parameter $\varepsilon$. The reference solution is the semi-analytical one described in Section \ref{subsec:validation}.}
\label{table:example 1-xi}
\begin{tabular}{c|cccc}
\hline
Error & $\varepsilon = 1/16$ & $\varepsilon = 1/32$ & $\varepsilon = 1/64$ & $\varepsilon = 1/128$ \\
\hline
$M = 100$ & 0.033264 & 0.033814 & 0.034086 & 0.034221 \\
$M = 200$ & 0.024683 & 0.025091 & 0.025293 & 0.025393 \\
$M = 400$ & 0.016555 & 0.016829 & 0.016964 & 0.017031 \\
$M = 800$ & 0.011587 & 0.011779 & 0.011874 & 0.011921 \\
$M = 1600$ & 0.008348 & 0.008486 & 0.0085541 & 0.008588 \\
\hline
\end{tabular}
\end{table}

\begin{figure}[htbp]
    \begin{subfigure}{0.48\textwidth}
        \centering
        \includegraphics[width=\textwidth]{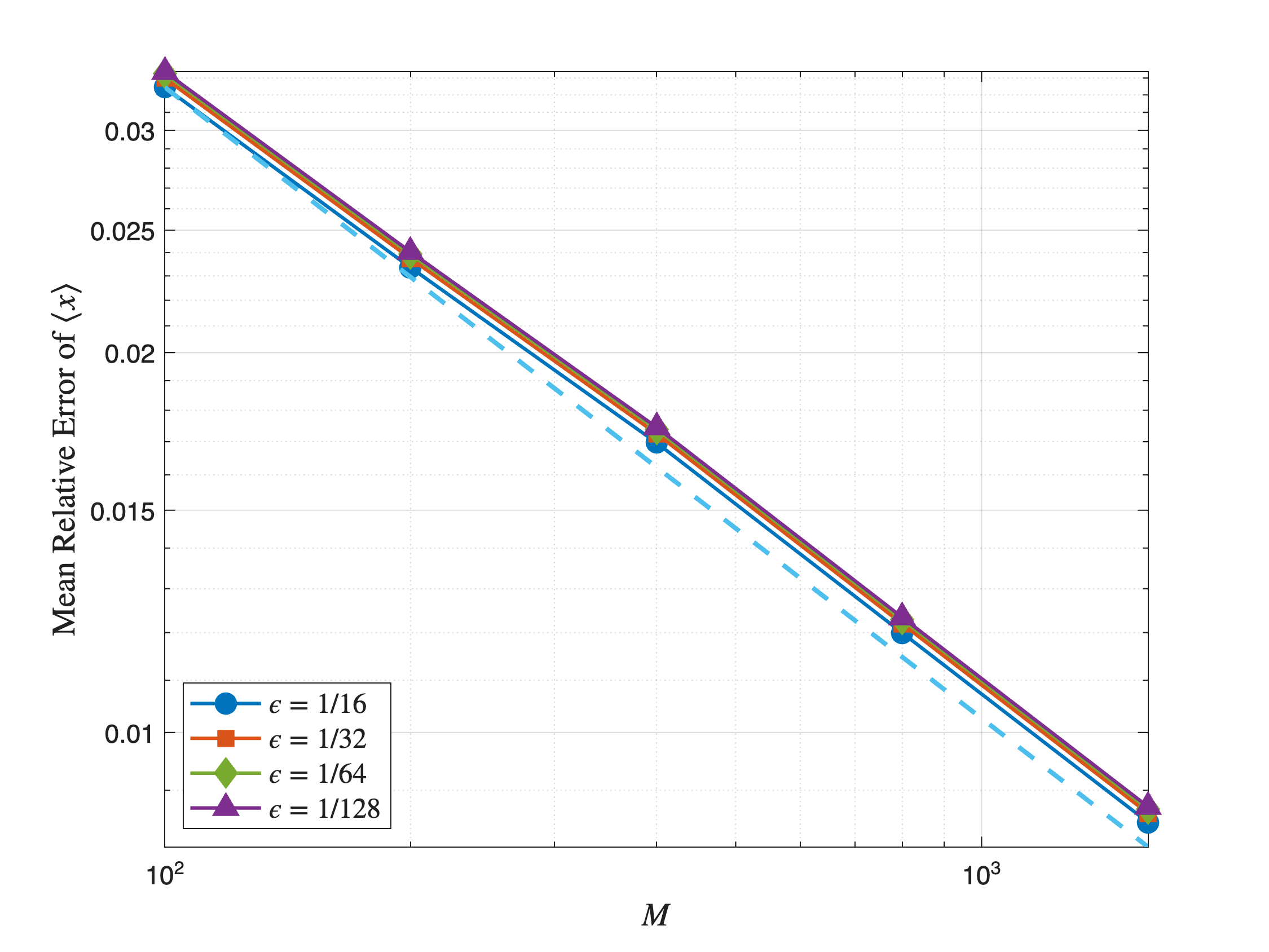}
        \label{fig:x of example 1}
    \end{subfigure}
    \hfill
    \begin{subfigure}{0.48\textwidth}
        \centering
        \includegraphics[width=\textwidth]{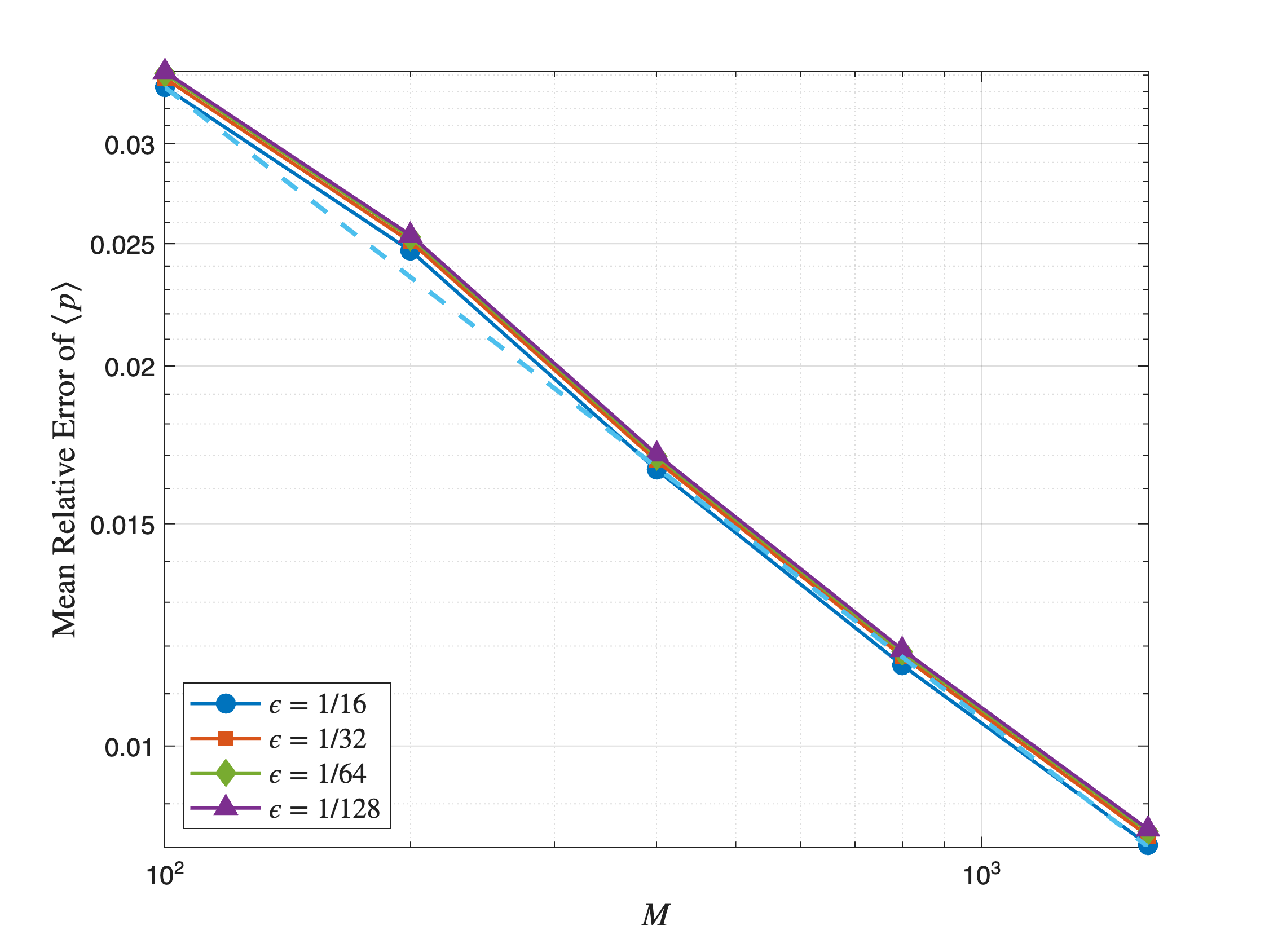}
        \label{fig:xi of example 1}
    \end{subfigure}
\caption{Log-log plot of the numerical errors versus the sample size $M$ for various values of $\varepsilon$ for Example 1 at $T=1.0$. (Left)Absolute error of the mean position $\langle x \rangle$. (Right) Absolute error of the mean momentum $\langle \xi \rangle$. The dashed lines with a slope of -1/2 are included to guide the eye, visually confirming the $\mathcal{O}(M^{-1/2})$ convergence rate.}
    \label{fig:example 1-log plot}
\end{figure}

\begin{figure}[htbp]
    \centering
    \includegraphics[width=\textwidth]{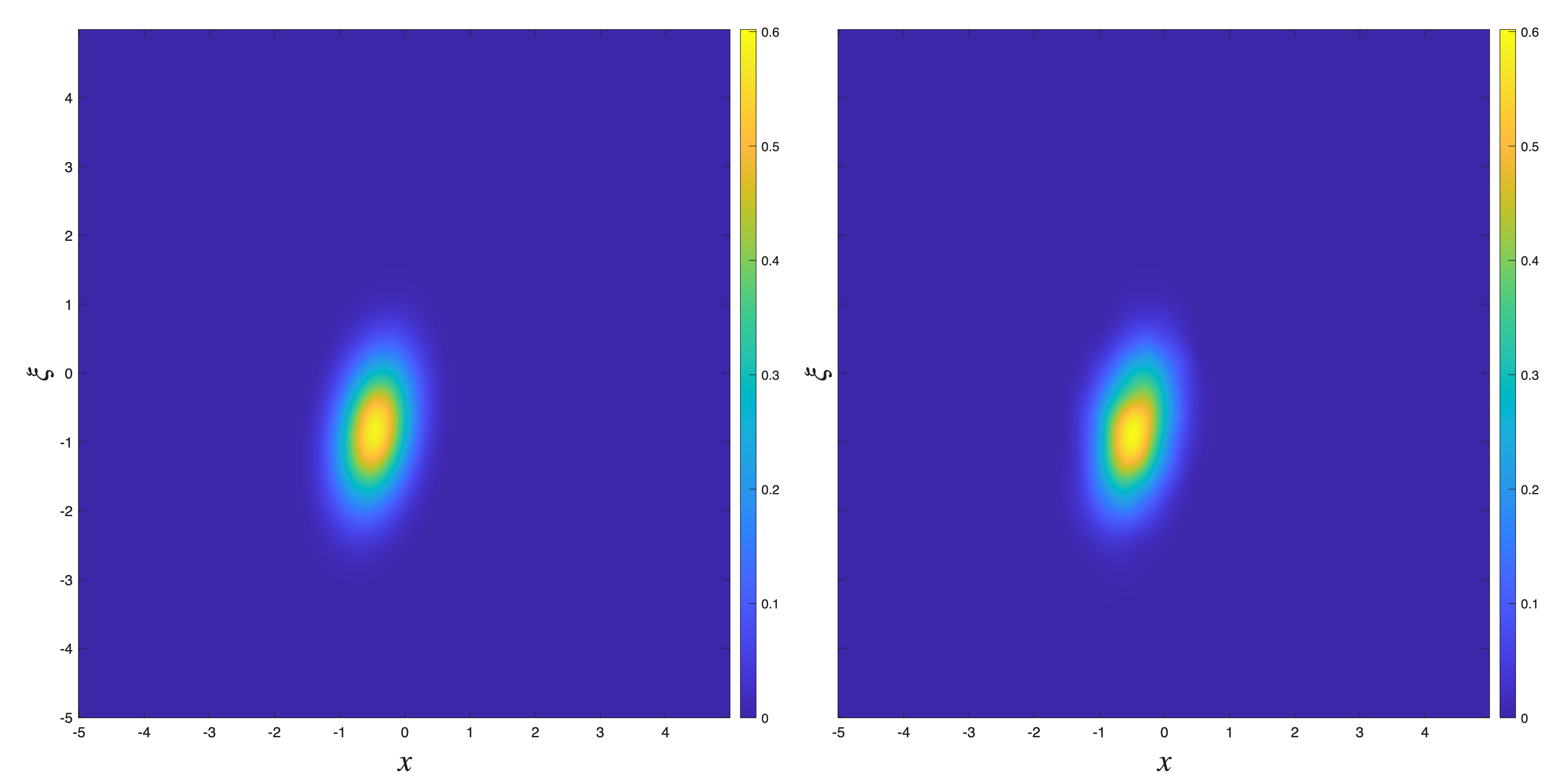}
\caption{Qualitative comparison of the Wigner function for Example 1 at $T=1.0$ with $\varepsilon=1/16$. (Left) The semi-analytical reference solution. (Right) The solution computed by the FGS algorithm with a sample size of $M=3200$. The excellent visual agreement demonstrates the algorithm's ability to capture the key phase space dynamics.}
    \label{fig:example 1}
\end{figure}

\subsubsection{Example 2: Double-well Potential}
Next, we test the algorithm on a non-harmonic, double-well potential, $V(x)=(x^2-1)^2$, with the initial wave packet centered at $(x_0, \xi_0) = (-0.1, 0.2)$. This example serves to demonstrate that the key advantages observed in the harmonic case—particularly the robustness to $\varepsilon$—are not an artifact of the simple potential but hold true for complex systems.

The numerical errors for the observables are documented in Tables \ref{table:example 2-x}-\ref{table:example 2-xi}. The data confirm that our primary finding from Example 1 persists: the algorithm's error for physical observables remains remarkably robust with respect to the semiclassical parameter $\varepsilon$, even in this non-harmonic case. As seen in the tables, for a fixed $M$, the error shows very weak dependence on $\varepsilon$.

Second, for a fixed $\varepsilon$, the errors also exhibit the expected $\mathcal{O}(M^{-1/2})$ convergence rate, consistent with the statistical nature of the method. These findings confirm that the key advantages of the FGS algorithm are general and not limited to quadratic potentials.

These two trends are visualized in the log-log plots in Figure \ref{fig:example 2-log plot}. The tight clustering of the error curves for different $\varepsilon$ values again visually confirms the algorithm's robustness to the semiclassical parameter, while the near-parallel slopes of the lines confirm the $\mathcal{O}(M^{-1/2})$ scaling. Qualitatively, Figure \ref{fig:example 2} compares the Wigner function computed by our FGS algorithm with the benchmark TSSP solution at $T=1.0$. The strong visual agreement demonstrates that our method successfully captures the more intricate phase space dynamics, such as the splitting and spiraling of the wavepacket, which are characteristic of this bistable potential.

\begin{table}[htbp]
\centering
\caption{Absolute errors of the mean position $\langle x\rangle$ for Example 2 at $T=1.0$. The table shows the errors as a function of the sample size $M$ and $\varepsilon$. The reference solution is the benchmark result from the TSSP method.}
\label{table:example 2-x}
\begin{tabular}{c|ccccc}
\hline
Error & $\varepsilon = 1/16$ & $\varepsilon = 1/32$ & $\varepsilon = 1/64$ & $\varepsilon = 1/128$ & $\varepsilon = 1/256$ \\
\hline
$M = 100$ & 0.044031 & 0.04456 & 0.04482 & 0.044949 & 0.045013 \\
$M = 200$ & 0.031018 & 0.031327 & 0.031483 & 0.031562 & 0.031602 \\
$M = 400$ & 0.022901 & 0.02315 & 0.023276 & 0.023338 & 0.023369 \\
$M = 800$ & 0.016363 &  0.016526 & 0.016613 & 0.016657 & 0.016679 \\
$M = 1600$ & 0.011563 & 0.011624 & 0.01176 & 0.011798 & 0.011817\\
\hline
\end{tabular}
\end{table}

\begin{table}[htbp]
\centering
\caption{Absolute errors of the mean momentum $\langle \xi\rangle$ for Example 2 at $T=1.0$. The table shows the errors as a function of the sample size $M$ and $\varepsilon$. The reference solution is the benchmark result from the TSSP method.}
\label{table:example 2-xi}
\begin{tabular}{c|ccccc}
\hline
Error & $\varepsilon = 1/16$ & $\varepsilon = 1/32$ & $\varepsilon = 1/64$ & $\varepsilon = 1/128$ & $\varepsilon = 1/256$ \\
\hline
$M = 100$ & 0.056284 & 0.056631 & 0.056801 & 0.056886 & 0.056928 \\
$M = 200$ & 0.038615 & 0.03883 & 0.038973 & 0.039054 & 0.039097 \\
$M = 400$ & 0.027825 & 0.027765 & 0.027785 & 0.027803 & 0.027814 \\
$M = 800$ & 0.020318 & 0.020341 & 0.020418 & 0.020471 & 0.020499 \\
$M = 1600$ & 0.014961 & 0.014822 & 0.014725 & 0.014716 & 0.014713 \\
\hline
\end{tabular}
\end{table}

\begin{figure}[htbp]
    \begin{subfigure}{0.48\textwidth}
        \centering
        \includegraphics[width=\textwidth]{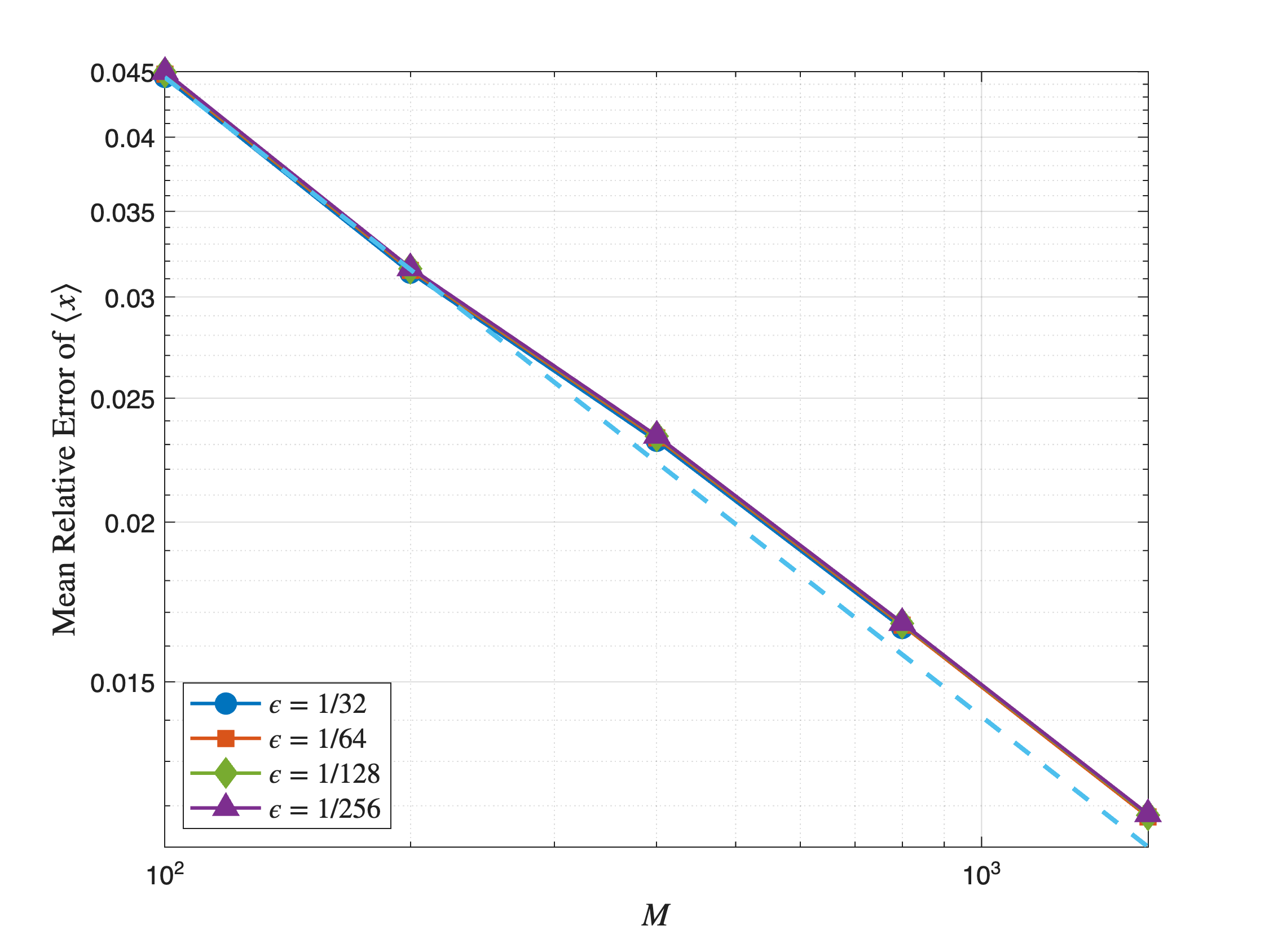}
        \label{fig:x of example 2}
    \end{subfigure}
    \hfill
    \begin{subfigure}{0.48\textwidth}
        \centering
        \includegraphics[width=\textwidth]{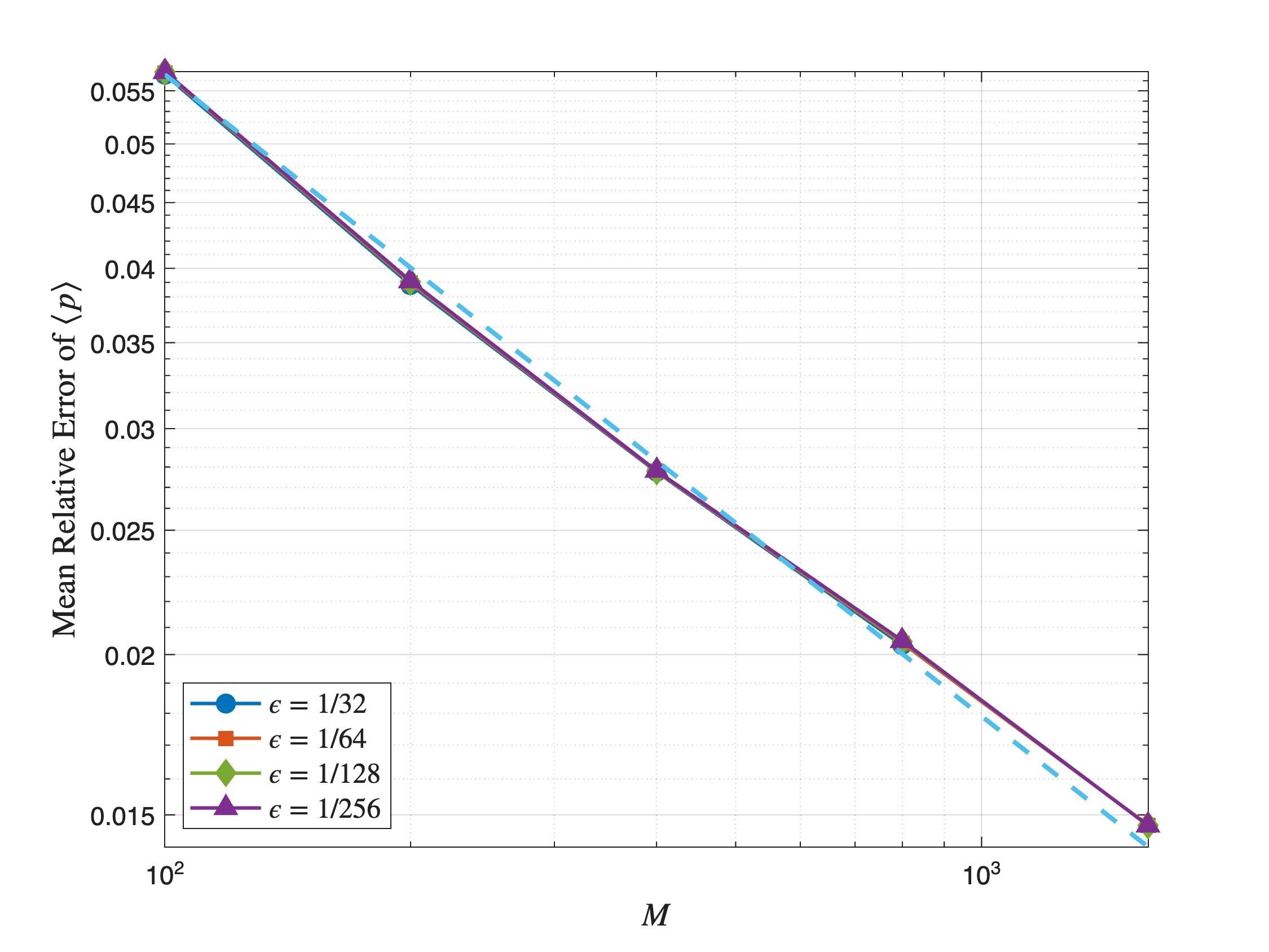}
        \label{fig:xi of example 2}
    \end{subfigure}
\caption{Log-log plot of the numerical errors versus the sample size $M$ for various values of $\varepsilon$ for Example 2 at $T=1.0$. (Left) Absolute error of the mean position $\langle x \rangle$. (Right) Absolute error of the mean momentum $\langle \xi \rangle$. The dashed lines with a slope of -1/2 visually confirm the consistent $\mathcal{O}(M^{-1/2})$ convergence rate.}
    \label{fig:example 2-log plot}
\end{figure}

\begin{figure}[htbp]
    \centering
    \includegraphics[width=\textwidth]{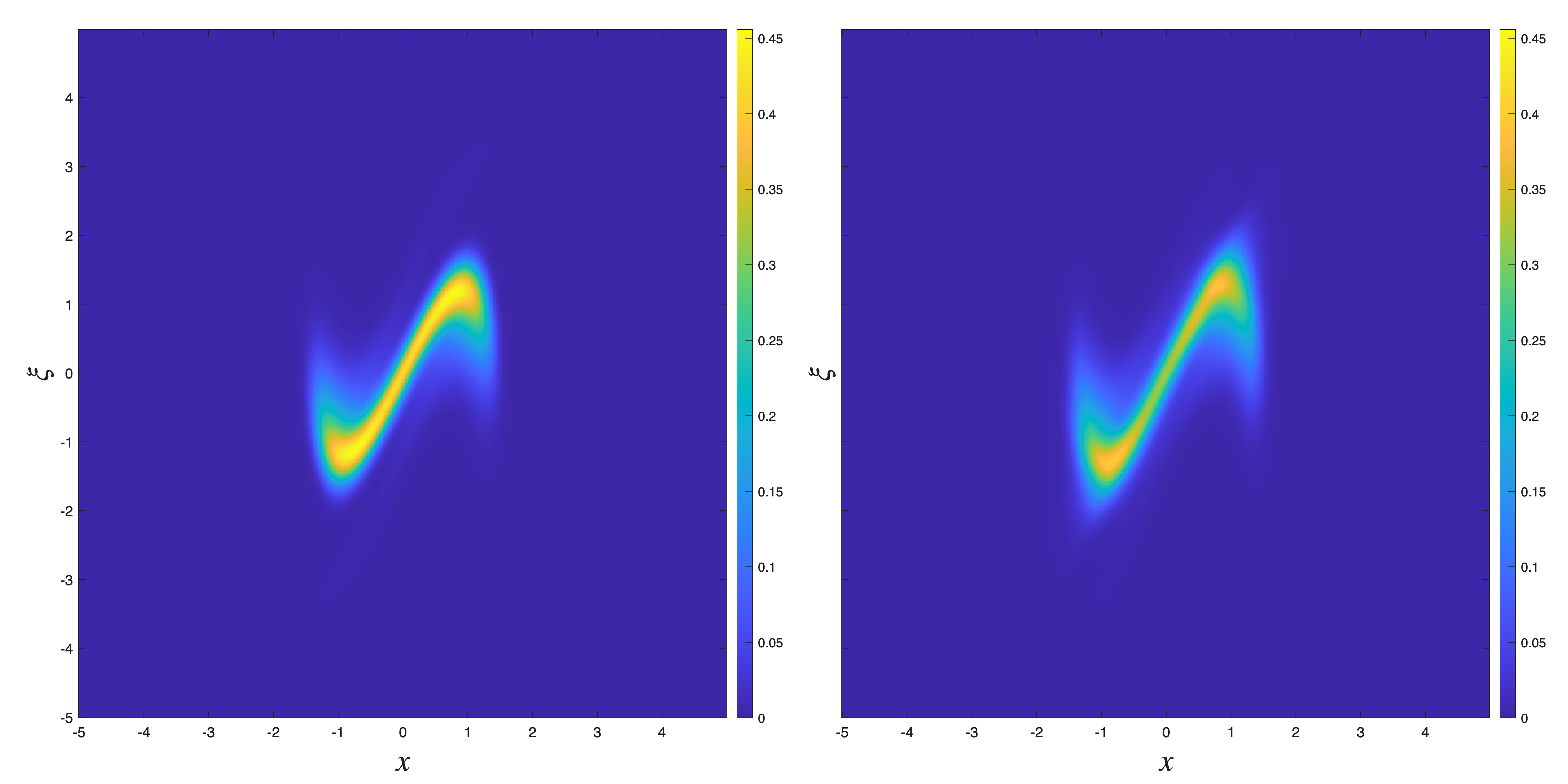}
\caption{Qualitative comparison of the Wigner function for Example 2 at $T=1.0$ with $\varepsilon=1/16$. (Left) The benchmark reference solution from TSSP. (Right) The solution computed by the FGS algorithm with a sample size of $M=3200$. The strong visual agreement highlights the algorithm's capability to handle complex dynamics.}
    \label{fig:example 2}
\end{figure}

\subsubsection{Example 3: Long-Time Dynamics and Stability}
Having established the accuracy and convergence of the FGS algorithm, we now demonstrate its superior stability for long-time simulations. This is a crucial advantage in scenarios where grid-based methods like TSSP are constrained by finite domain sizes. We consider the double-well potential $V(x)=(x^2-1)^2$ with an initial wave packet centered at $(x_0, \xi_0) = (-0.2, 0.2)$ and $\varepsilon=1/16$.

The results of this comparison are presented in Figure \ref{fig: example 3}. The top row shows that the TSSP method, on a computational domain of $[-5,5]^2$, accurately captures the dynamics up to $T=4$. However, by $T=8$ (middle left panel), the wave packet has spread to the boundaries of this domain. Consequently, the TSSP solution exhibits significant spurious artifacts due to boundary effects, rendering the result unreliable.

The middle right panel demonstrates that merely refining the grid resolution (in this case, to $\mathrm{d}x = \mathrm{d}\xi = 0.001$) does not resolve this issue, confirming that the problem stems from the limited domain size, not insufficient grid resolution. Indeed, the TSSP method can recover a valid solution at $T=8$ only if the computational domain is significantly enlarged to $[-8,8]^2$ (bottom left panel). This series of tests clearly illustrates the inherent limitation of fixed-grid methods for long-time simulations where the solution may explore a large region of phase space.

In stark contrast, the bottom right panel shows the result from our FGS algorithm. Since the FGS method evolves particles in an unbounded phase space, it is inherently immune to such boundary-induced instabilities. The solution remains stable and accurately captures the complex dynamics at $T=8$ without requiring any domain adjustments. This example highlights a key practical advantage of the FGS algorithm: its robustness and reliability for simulating long-time quantum dynamics.

\begin{figure}[htbp]
    \centering
    \begin{subfigure}{0.48\textwidth}
        \centering
        \includegraphics[width=\textwidth]{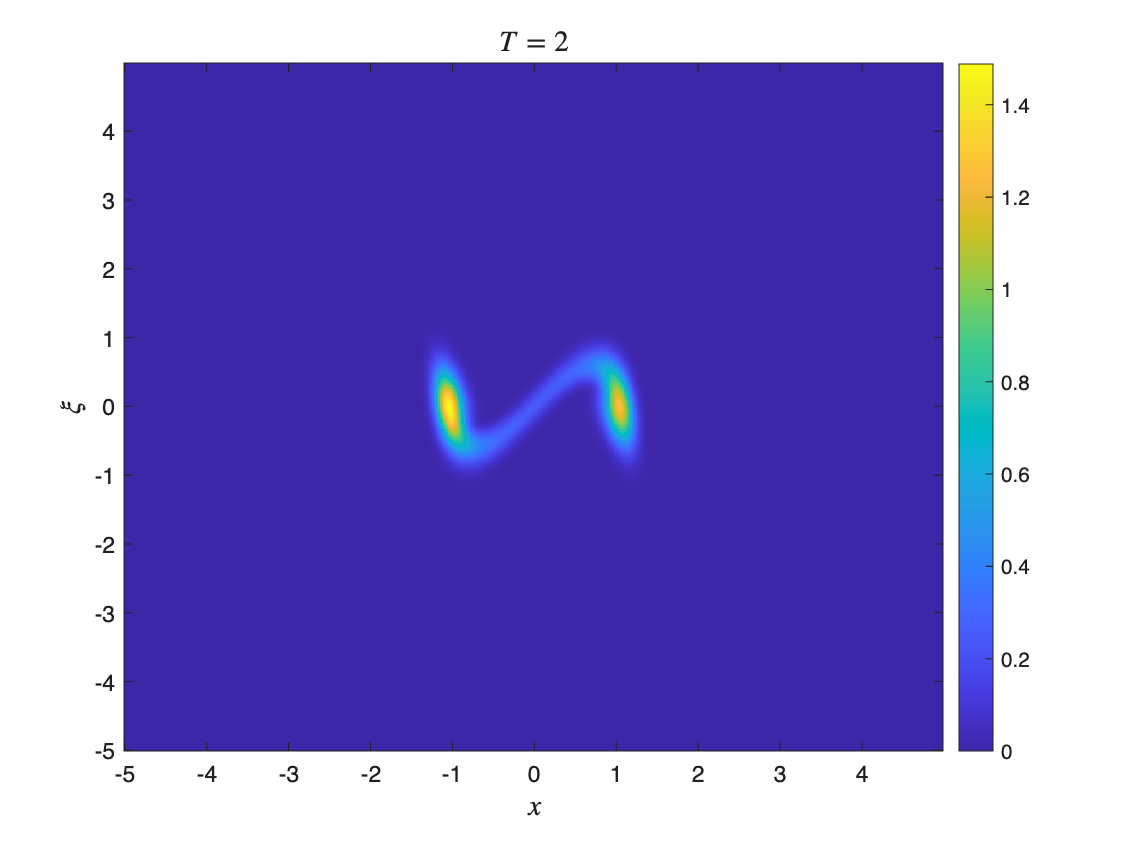}
        \label{fig: time 2 of example 3}
    \end{subfigure}
    \hfill
    \begin{subfigure}{0.48\textwidth}
        \centering
        \includegraphics[width=\textwidth]{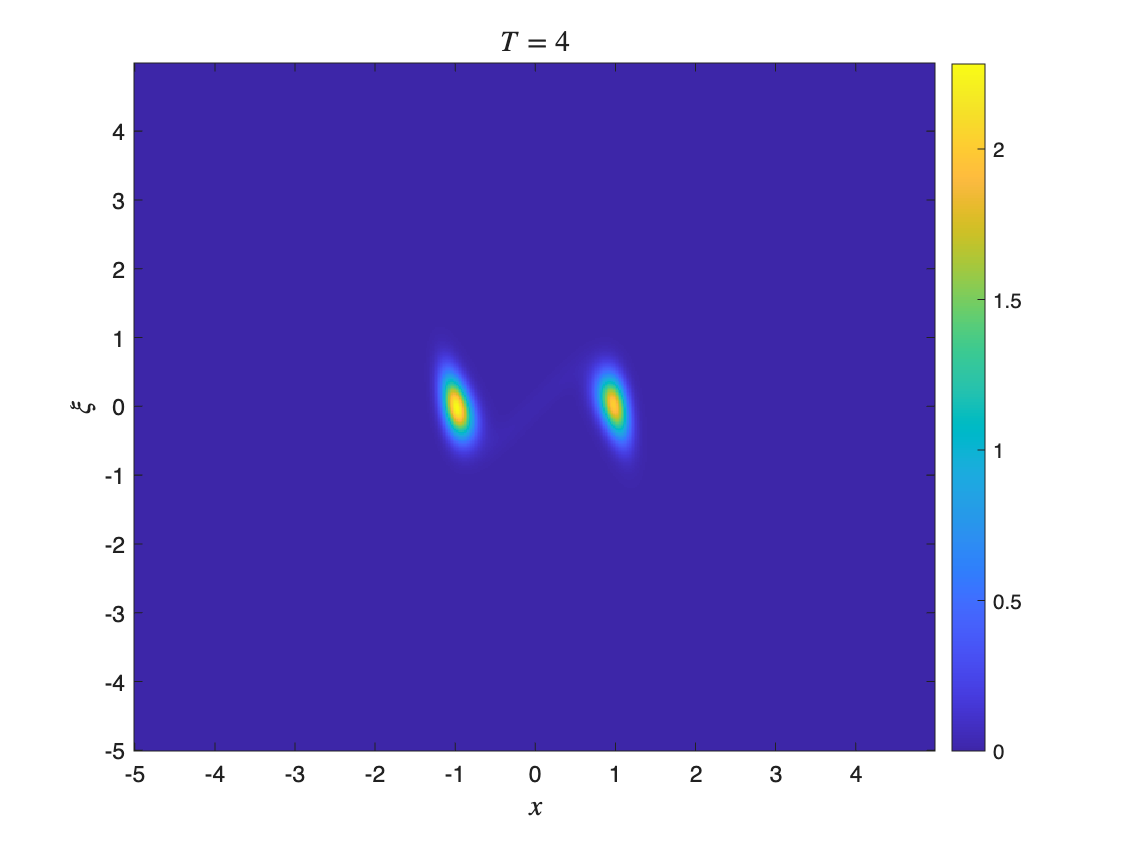}
        \label{fig: time 4 of example 3}
    \end{subfigure}
    \\
    \begin{subfigure}{0.48\textwidth}
        \centering
        \includegraphics[width=\textwidth]{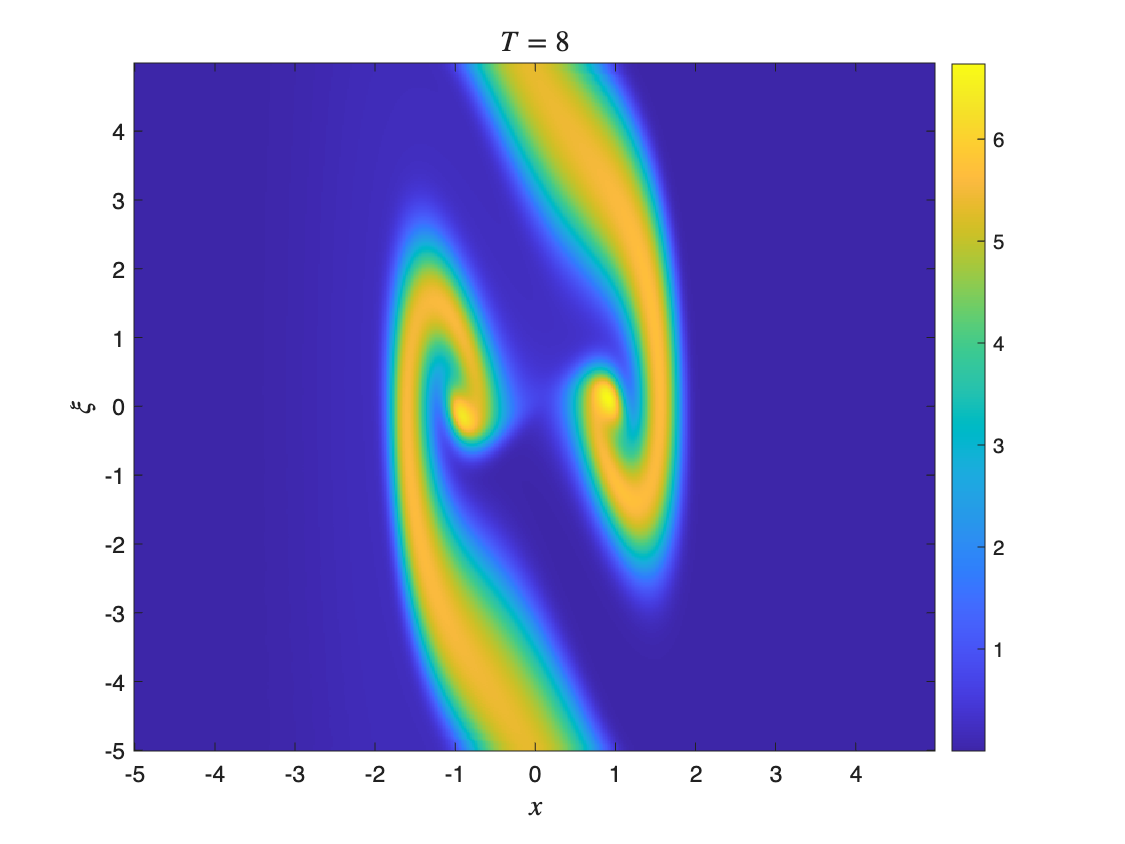}
        \label{fig: time 8 of example 3}
    \end{subfigure}
    \hfill
    \begin{subfigure}{0.48\textwidth}
        \centering
        \includegraphics[width=\textwidth]{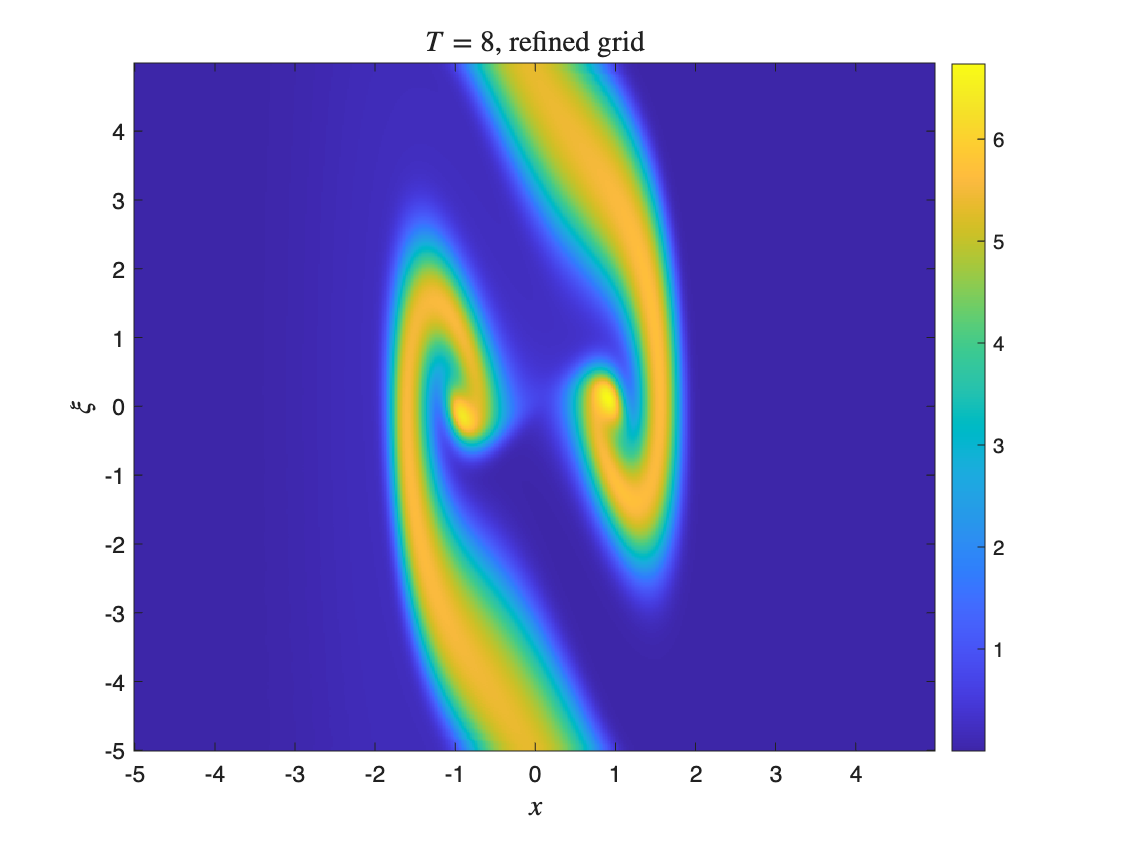} 
        \label{fig: time 8 for refined grid of example 3}
    \end{subfigure}
    \\
    \begin{subfigure}{0.48\textwidth}
        \centering
        \includegraphics[width=\textwidth]{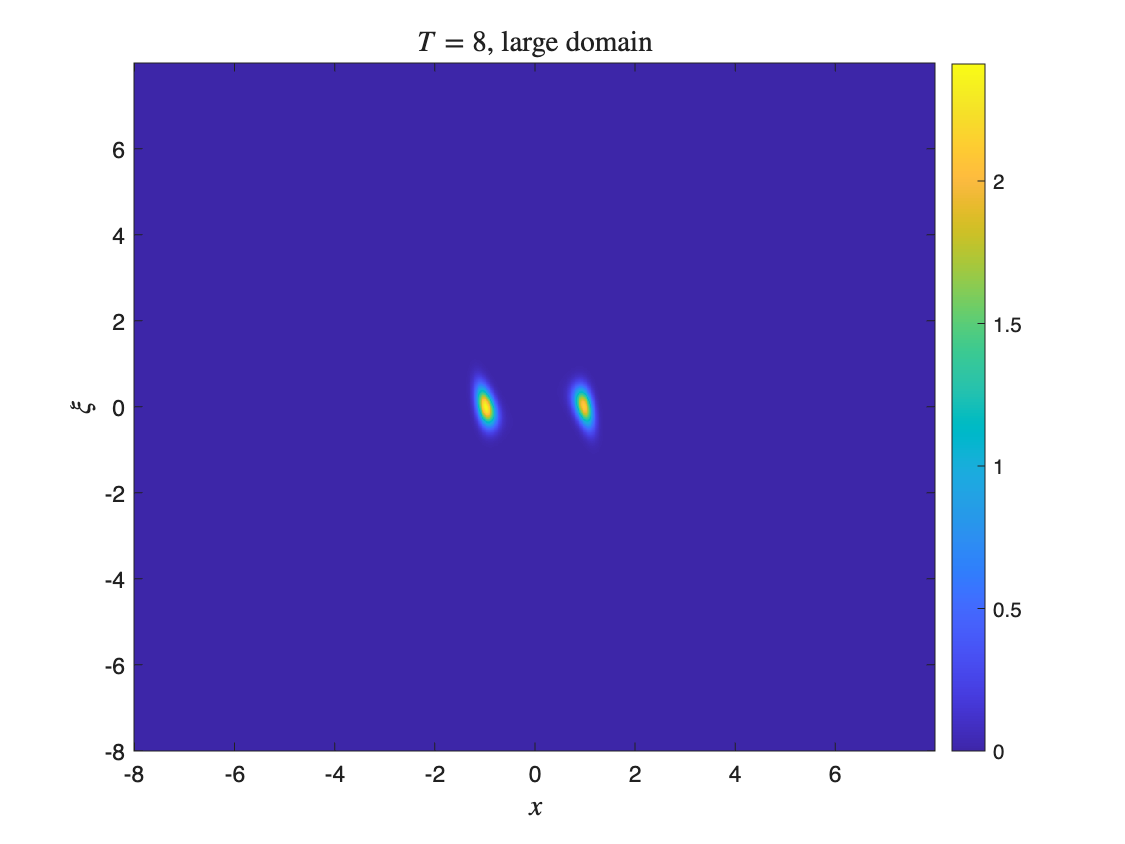}
        \label{fig: time 8 for large domain of example 3}
    \end{subfigure}
    \hfill
    \begin{subfigure}{0.48\textwidth}
        \centering
        \includegraphics[width=\textwidth]{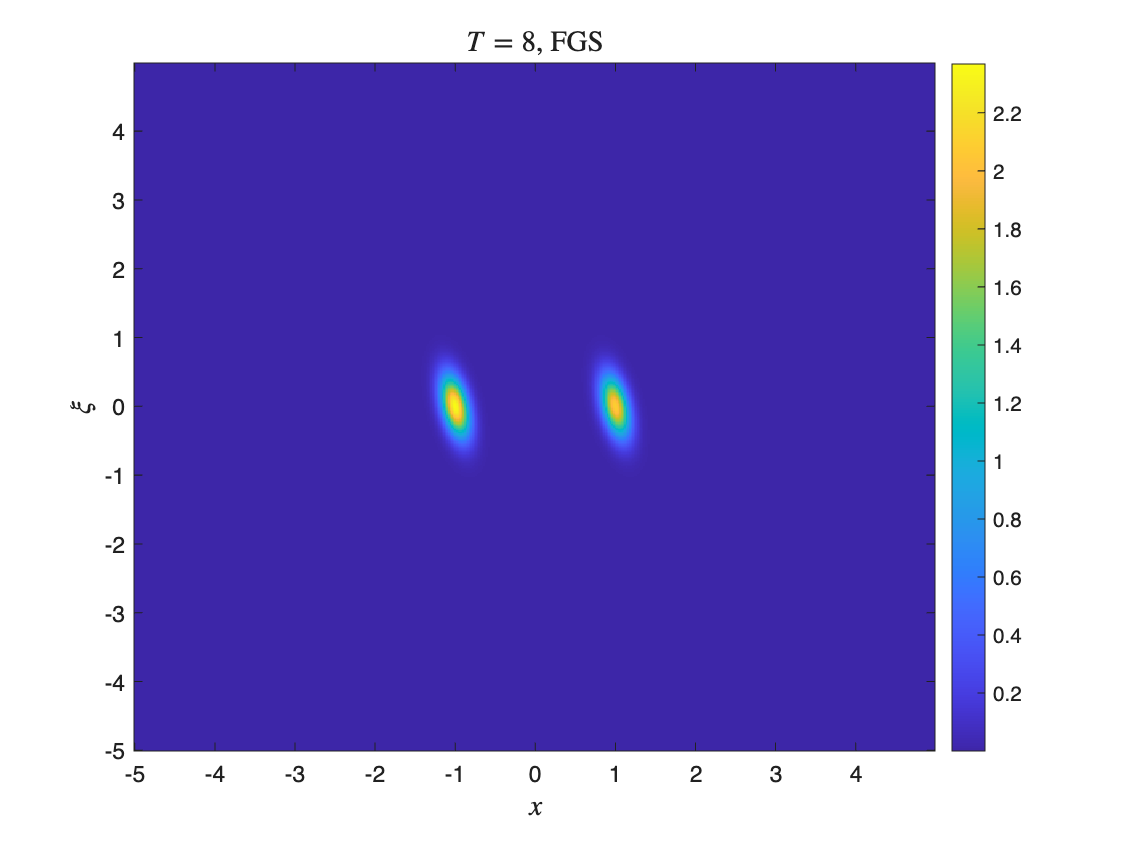}
        \label{fig: time 8 for sampling of example 3}
    \end{subfigure}
\caption{Demonstration of the FGS algorithm's stability for long-time simulation in Example 3 ($\varepsilon=1/16$). All simulations except the bottom right panel are performed with the TSSP method.
(Top row) TSSP solutions at $T=2$ (left) and $T=4$ (right) on a domain of $[-5,5]^2$.
(Middle row) At $T=8$, the TSSP solution on the original domain is corrupted by boundary artifacts (left). Refining the grid to $\mathrm{d}x = \mathrm{d}\xi = 0.001$ and $\mathrm{d}t = 0.0005$ does not fix the issue (right).
(Bottom row) The TSSP solution is recovered at $T=8$ by using a larger domain of $[-8,8]^2$ (left). In contrast, the FGS algorithm ($M=3200$) provides a stable and accurate result without domain limitations (right). The figure highlights the superior stability of the mesh-free FGS algorithm for long-time evolution.}
    \label{fig: example 3}
\end{figure}

\subsection{Numerical Explorations of Physical Phenomena}
Having validated the FGS algorithm and demonstrated its stability, we now apply it as an investigatory tool to explore a fundamental physical question pertinent to the Wigner-Fokker-Planck equation: the existence and nature of steady states for non-harmonic potentials.

\subsubsection{Example 4: Relaxation toward a Steady State}
In this example, we leverage the long-time stability of the FGS algorithm to numerically investigate the existence of a steady state for the Wigner-Fokker-Planck equation, particularly for non-harmonic potentials where theoretical results are scarce. To this end, we consider the following three distinct potentials, also shown in Figure \ref{fig:potential of example 5}:
\begin{align}
V_1(x) &= x^2+x+\sin(x), \label{pot: potential_1 in example 5} \\
V_2(x) &= (x^2-1)^2, \label{pot: potential_2 in example 5} \\
V_3(x) &= 0.08x^2(x^2-2)^2. \label{pot: potential_3 in example 5}
\end{align}
The first potential \eqref{pot: potential_1 in example 5} is a near-harmonic potential representing a small perturbation of a quadratic trap. The second \eqref{pot: potential_2 in example 5} is a bistable, double-well potential commonly used in studies of non-adiabatic quantum dynamics. The third potential \eqref{pot: potential_3 in example 5} is a symmetric \textit{triple-well} potential, governed by a sextic polynomial ($\sim x^6$). This model, often relevant in the study of structural phase transitions ($\phi^6$ model), presents a landscape with three stable minima separated by barriers.

The long-time evolution of the Wigner function under these three potentials is presented in Figures \ref{fig: steady state of potential 1 in example 5}, \ref{fig: steady state of potential 2 in example 5}, and \ref{fig: steady state of potential 3 in example 5}, respectively. In all three cases, despite the diverse potential landscapes, the Wigner function evolves and clearly converges to a time-independent profile. For instance, the distributions at $T=15$ and $T=20$ are visually indistinguishable, indicating that a steady state has been reached.

These numerical results carry significant implications. For the near-harmonic potential \eqref{pot: potential_1 in example 5}, our findings are consistent with the analytical work in \cite{arnold2012wigner}, which proves the existence of a steady state for potentials that are small perturbations of a harmonic one. Crucially, for the strongly non-harmonic potentials \eqref{pot: potential_2 in example 5} and \eqref{pot: potential_3 in example 5}, our simulations provide strong numerical evidence for the existence of a unique steady state. To the best of our knowledge, this is a regime where a rigorous analytical proof is still an open problem. Our work thus demonstrates the utility of the FGS algorithm as a powerful tool for exploring fundamental questions in open quantum systems that are currently beyond the reach of analytical theory.

\begin{figure}[htbp]
    \centering
    \begin{subfigure}{0.48\textwidth}
        \centering
        \includegraphics[width=\textwidth]{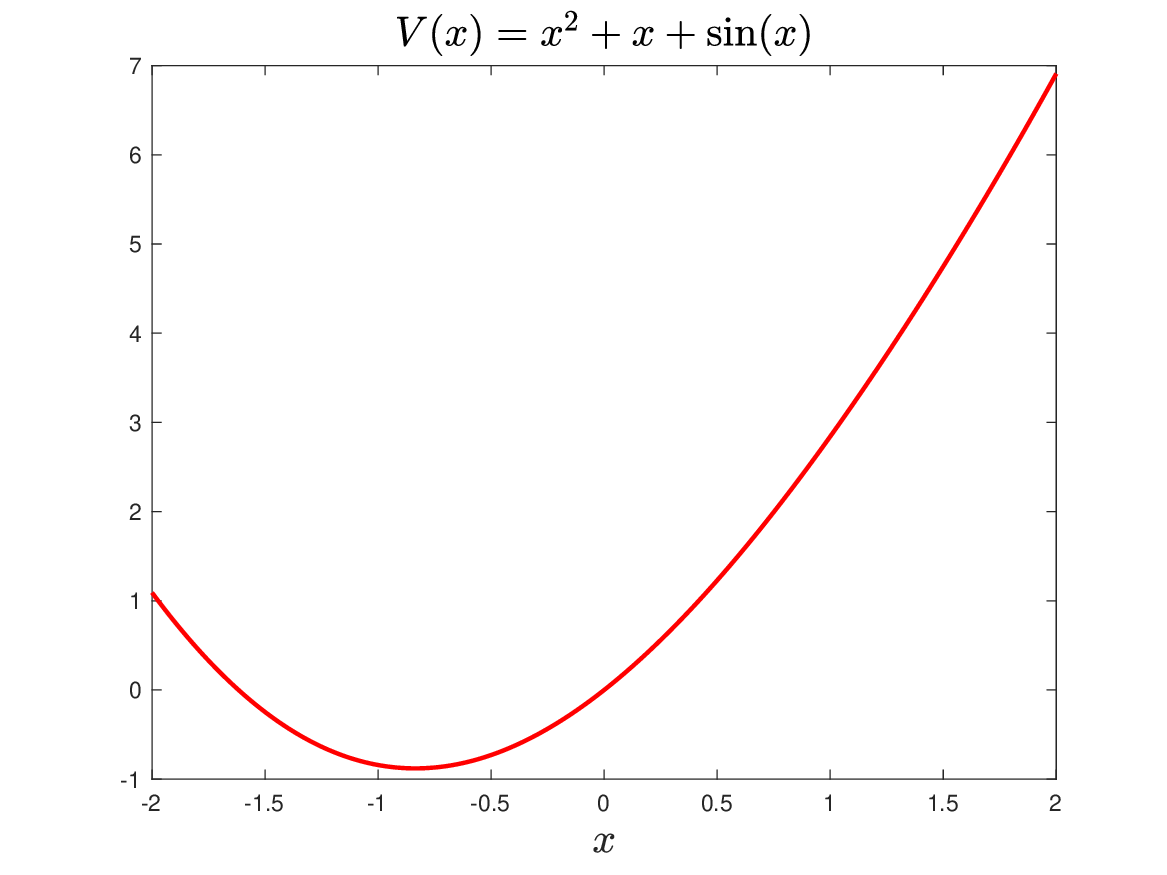}
        \caption{}
    \end{subfigure}
    \hfill
    \begin{subfigure}{0.48\textwidth}
        \centering
        \includegraphics[width=\textwidth]{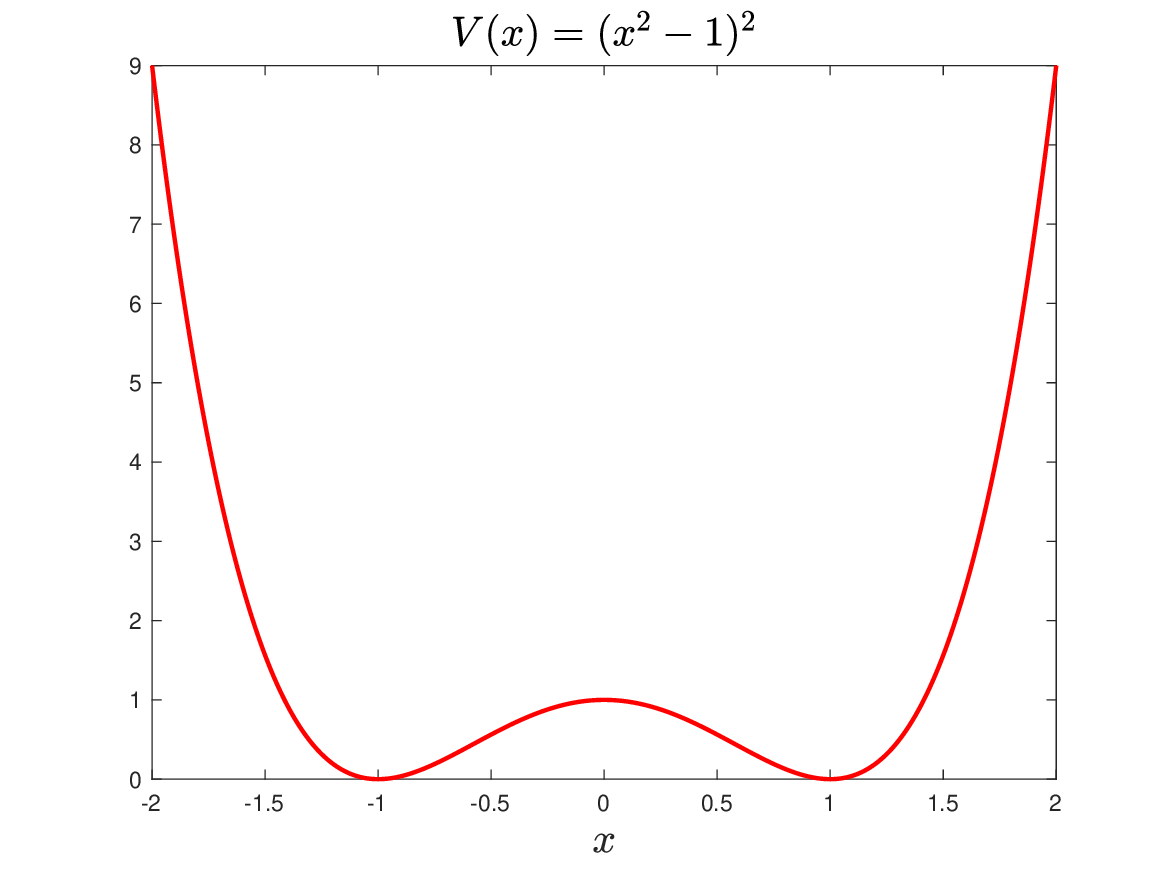}
        \caption{}
    \end{subfigure}
    \begin{subfigure}{0.48\textwidth}
        \centering
        \includegraphics[width=\textwidth]{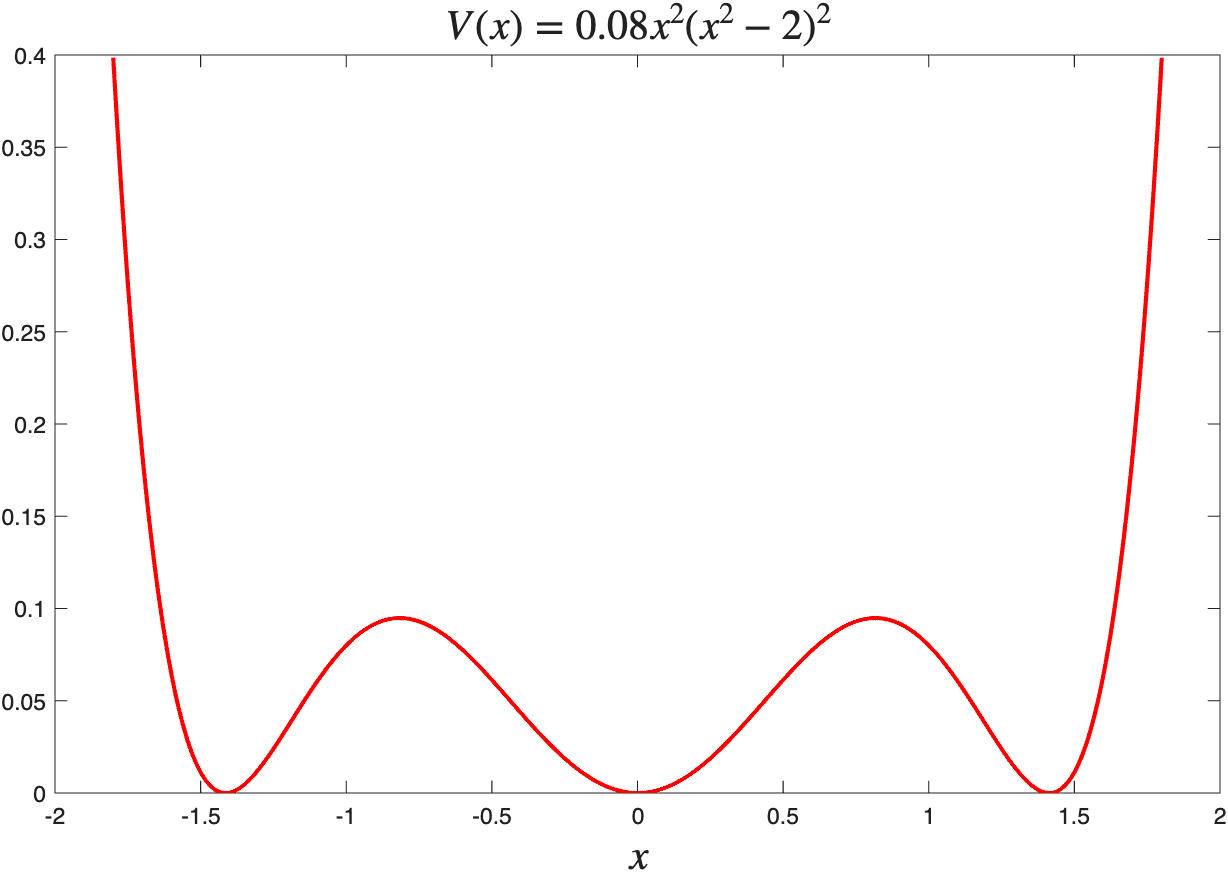}
        \caption{}
    \end{subfigure}
\caption{The three non-harmonic potentials investigated in Example 5. (a) A near-harmonic potential, Eq. \eqref{pot: potential_1 in example 5}. (b) A double-well bistable potential, Eq. \eqref{pot: potential_2 in example 5}. (c) A symmetric triple-well potential, Eq. \eqref{pot: potential_3 in example 5}, featuring three distinct local minima.}
\label{fig:potential of example 5}
\end{figure}

\begin{figure}[htbp]
    \centering
    \begin{subfigure}{0.48\textwidth}
        \centering
        \includegraphics[width=\textwidth]{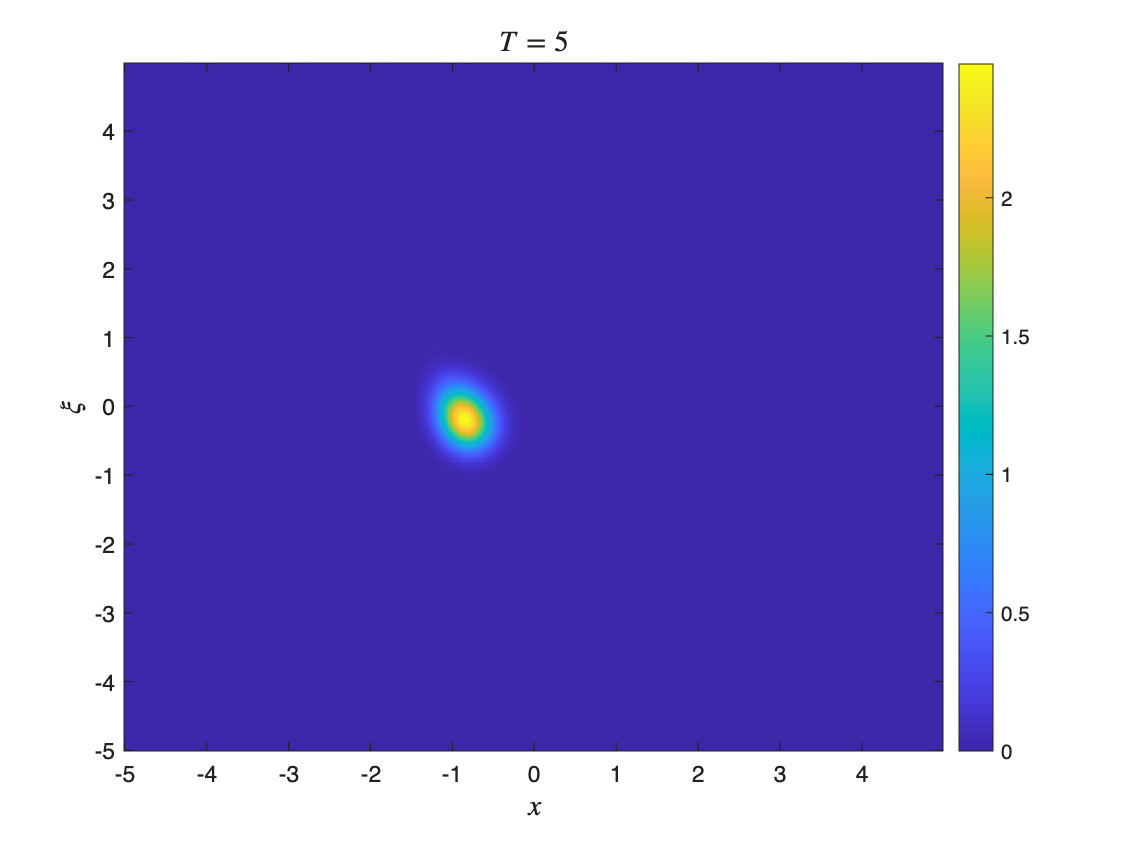}
        \label{fig: time 5 of potential 1 in example 5 }
    \end{subfigure}
    \hfill
    \begin{subfigure}{0.48\textwidth}
        \centering
        \includegraphics[width=\textwidth]{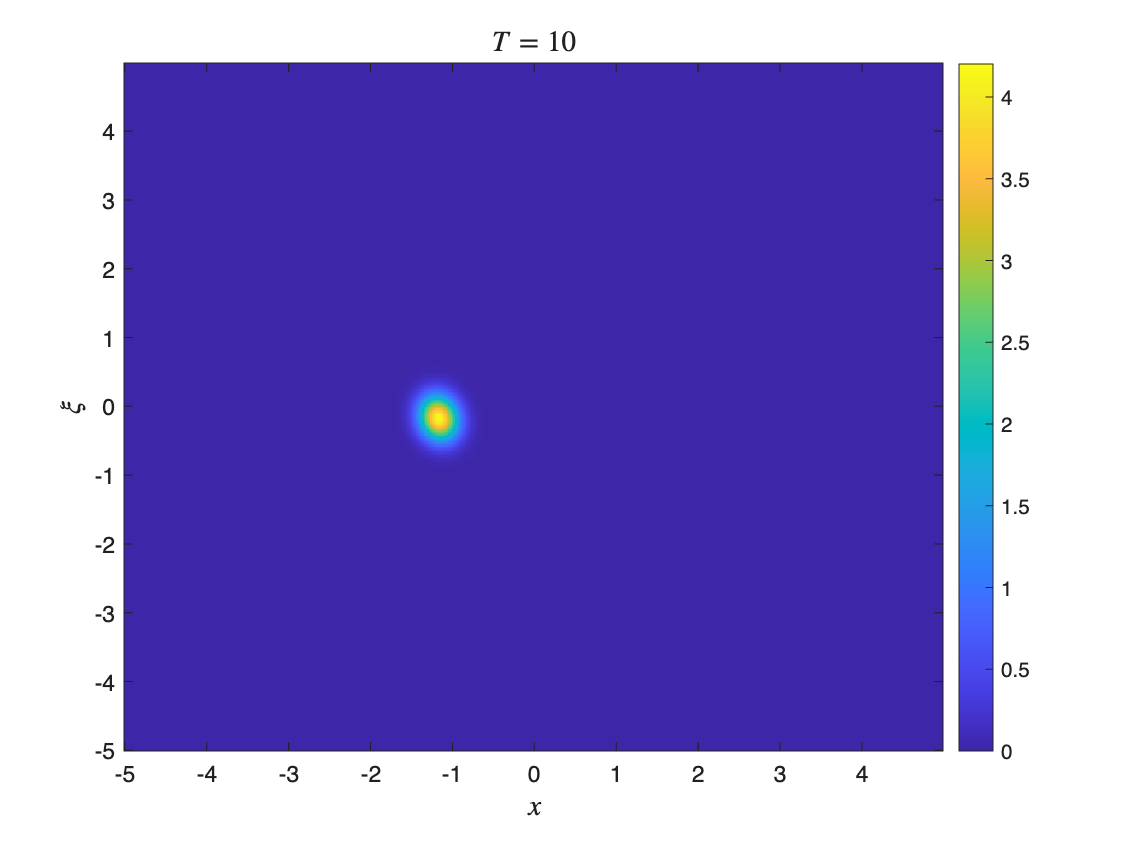}
        \label{fig: time 10 of potential 1 in example 5}
    \end{subfigure}
    \\
    \begin{subfigure}{0.48\textwidth}
        \centering
        \includegraphics[width=\textwidth]{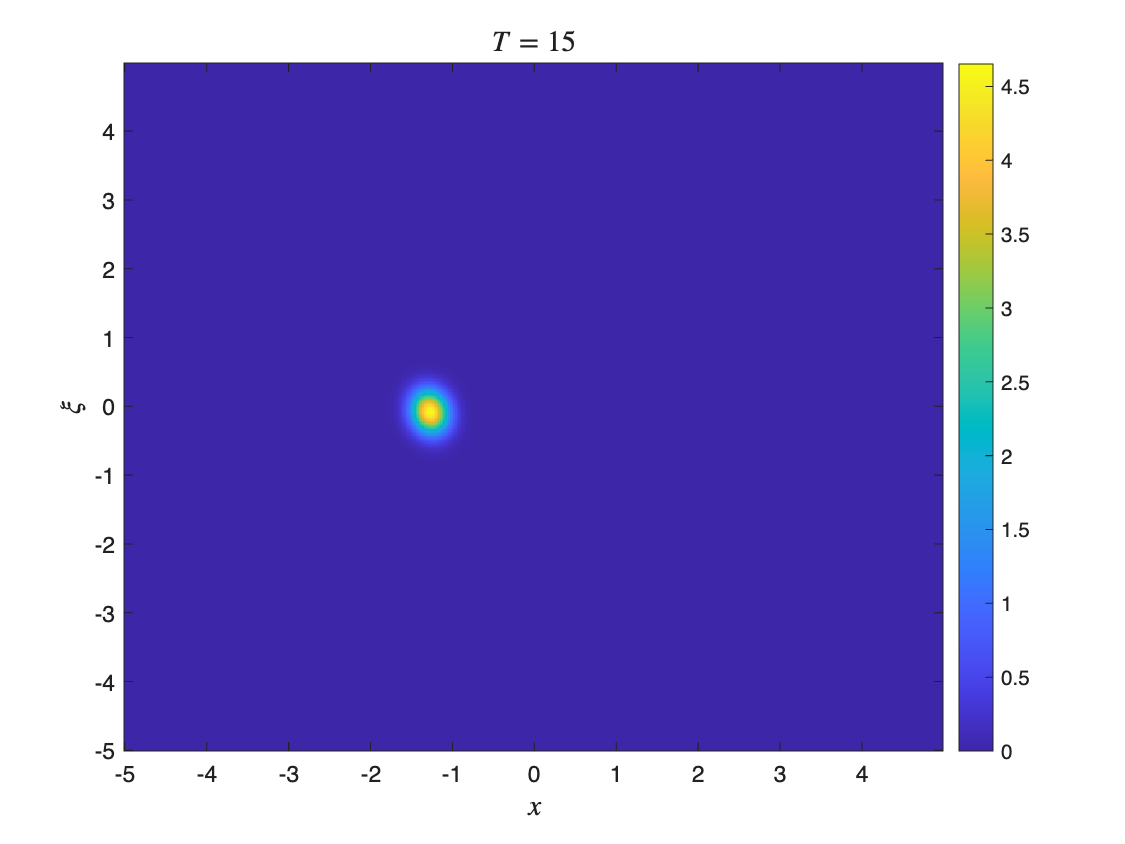}
        \label{fig: time 15 of potential 1 in example 5}
    \end{subfigure}
    \hfill
    \begin{subfigure}{0.48\textwidth}
        \centering
        \includegraphics[width=\textwidth]{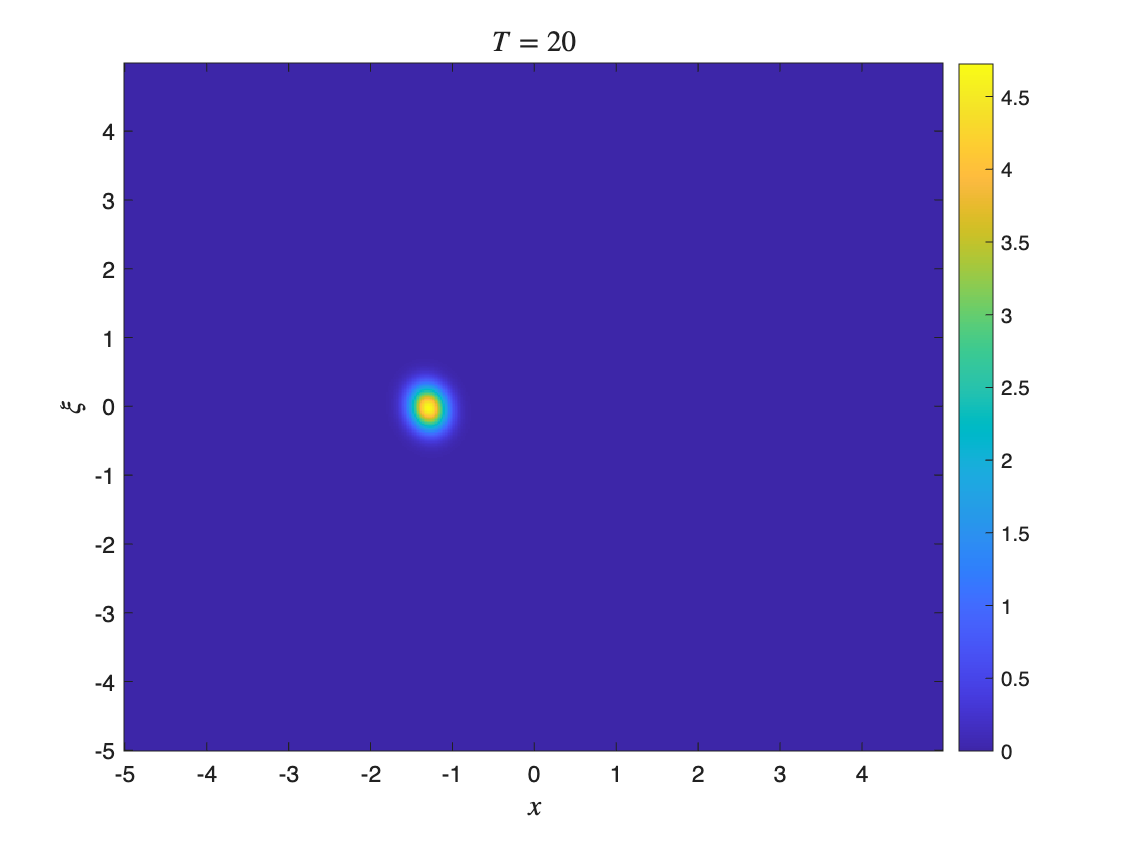}
        \label{fig: time 20 of potential 1 in example 5}
    \end{subfigure}
\caption{Long-time evolution of the Wigner function towards a steady state for Example 4, under the near-harmonic potential \eqref{pot: potential_1 in example 5} with $\varepsilon=1/32$. From left to right and top to bottom, the plots correspond to times $T=5, 10, 15,$ and $20$. The distribution is visually converged by $T=15$.}
    \label{fig: steady state of potential 1 in example 5}
\end{figure}

\begin{figure}[htbp]
    \centering
    \begin{subfigure}{0.48\textwidth}
        \centering
        \includegraphics[width=\textwidth]{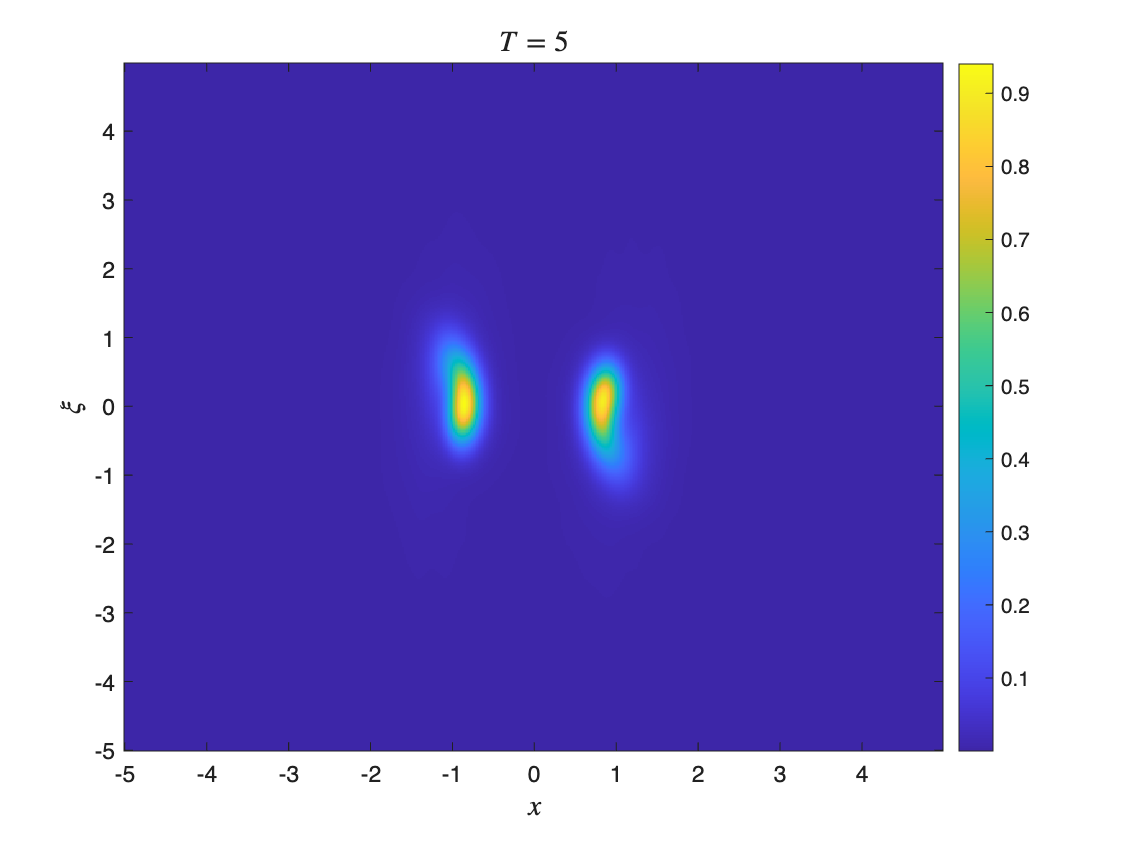}
        \label{fig: time 5 of potential 2 in example 5 }
    \end{subfigure}
    \hfill
    \begin{subfigure}{0.48\textwidth}
        \centering
        \includegraphics[width=\textwidth]{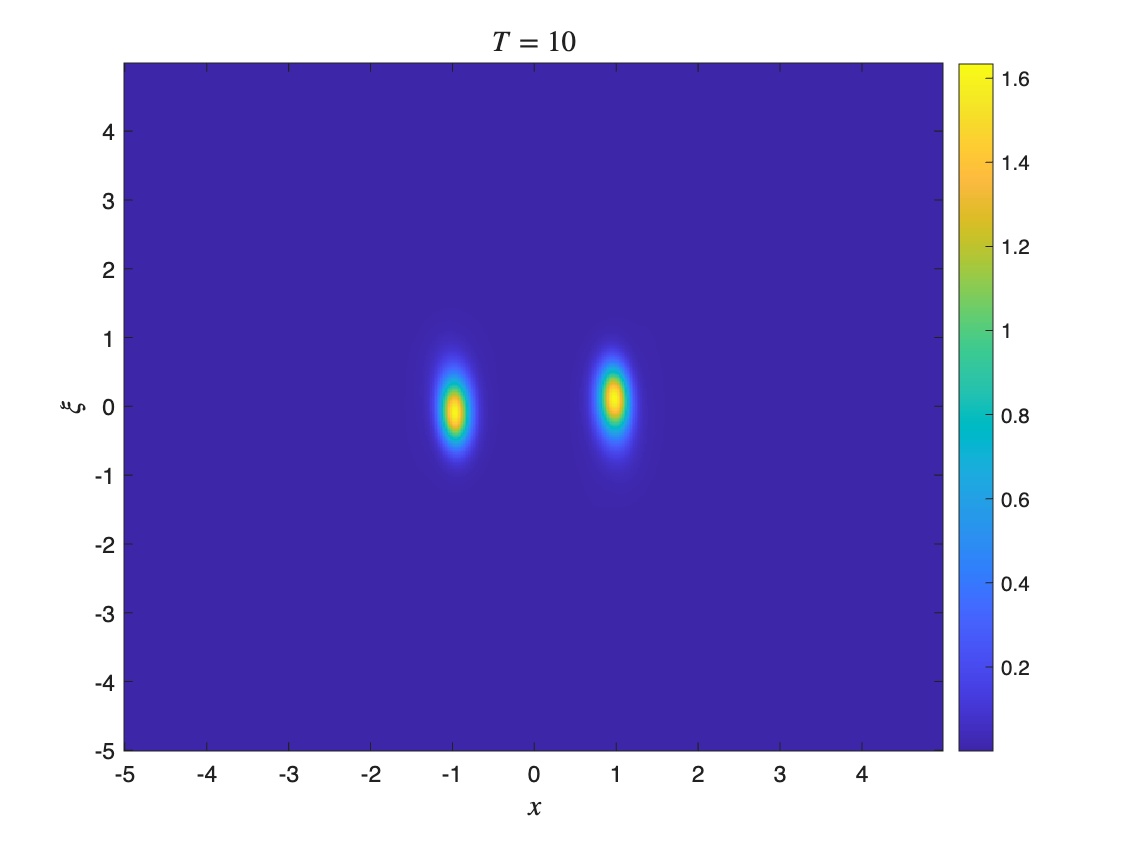}
        \label{fig: time 10 of potential 2 in example 5}
    \end{subfigure}
    \\
    \begin{subfigure}{0.48\textwidth}
        \centering
        \includegraphics[width=\textwidth]{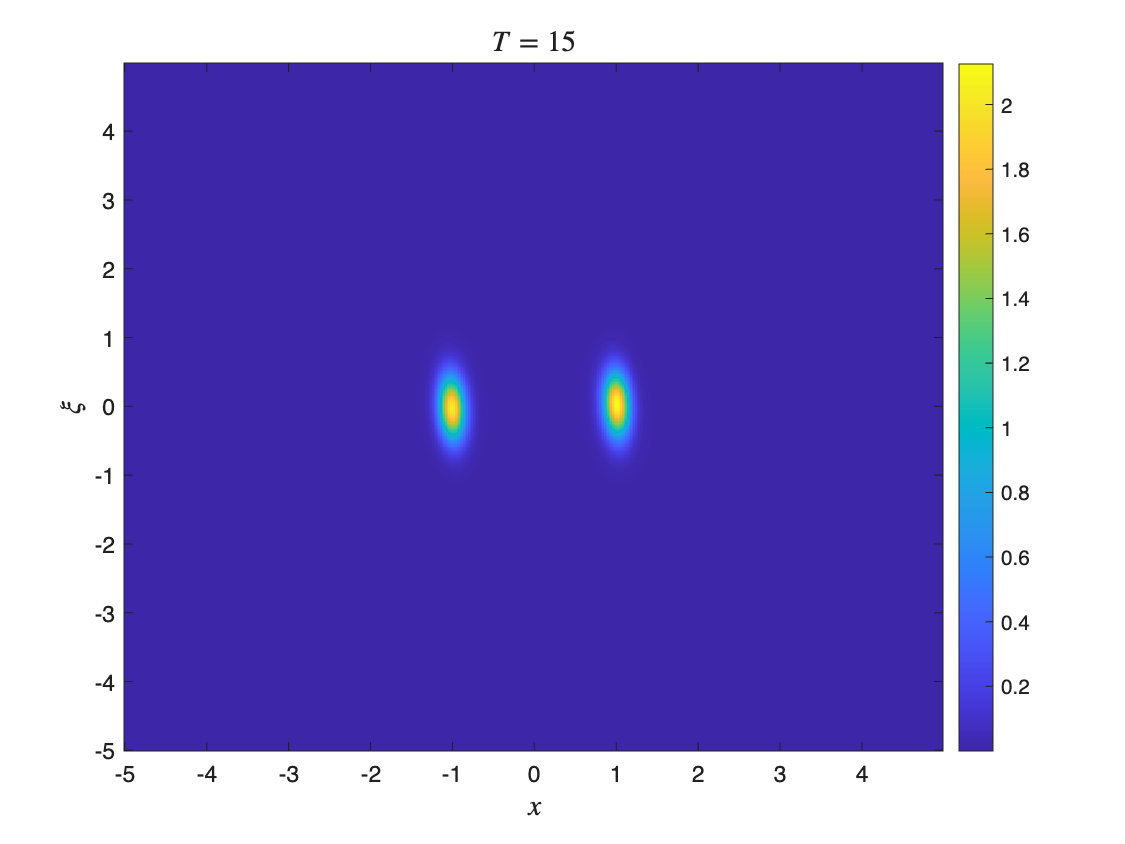}
        \label{fig: time 15 of potential 2 in example 5}
    \end{subfigure}
    \hfill
    \begin{subfigure}{0.48\textwidth}
        \centering
        \includegraphics[width=\textwidth]{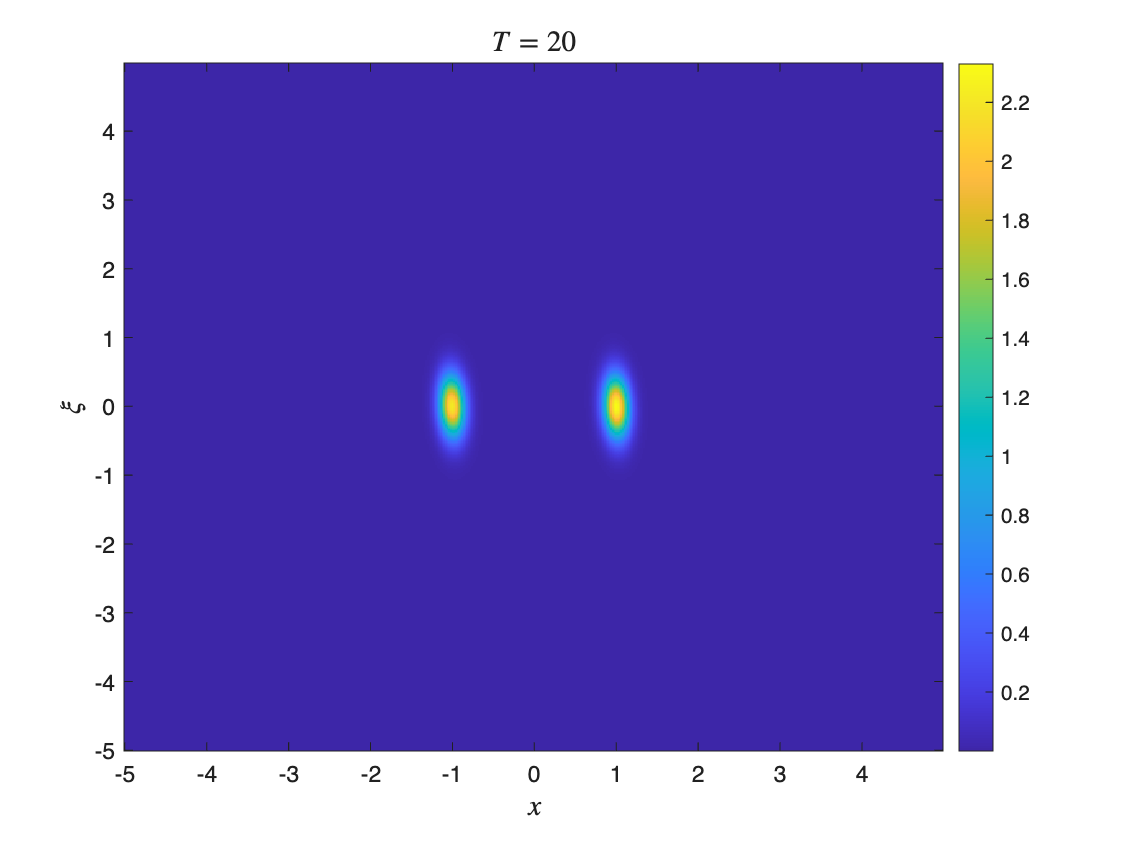}
        \label{fig: time 20 of potential 2 in example 5}
    \end{subfigure}
\caption{Long-time evolution of the Wigner function towards a steady state for Example 4, under the double-well potential \eqref{pot: potential_2 in example 5} with $\varepsilon=1/32$. From left to right and top to bottom, the plots correspond to times $T=5, 10, 15,$ and $20$. The distribution is visually converged by $T=15$.}
    \label{fig: steady state of potential 2 in example 5}
\end{figure}

\begin{figure}[htbp]
    \centering
    \begin{subfigure}{0.49\textwidth}
        \centering
        \includegraphics[width=\textwidth]{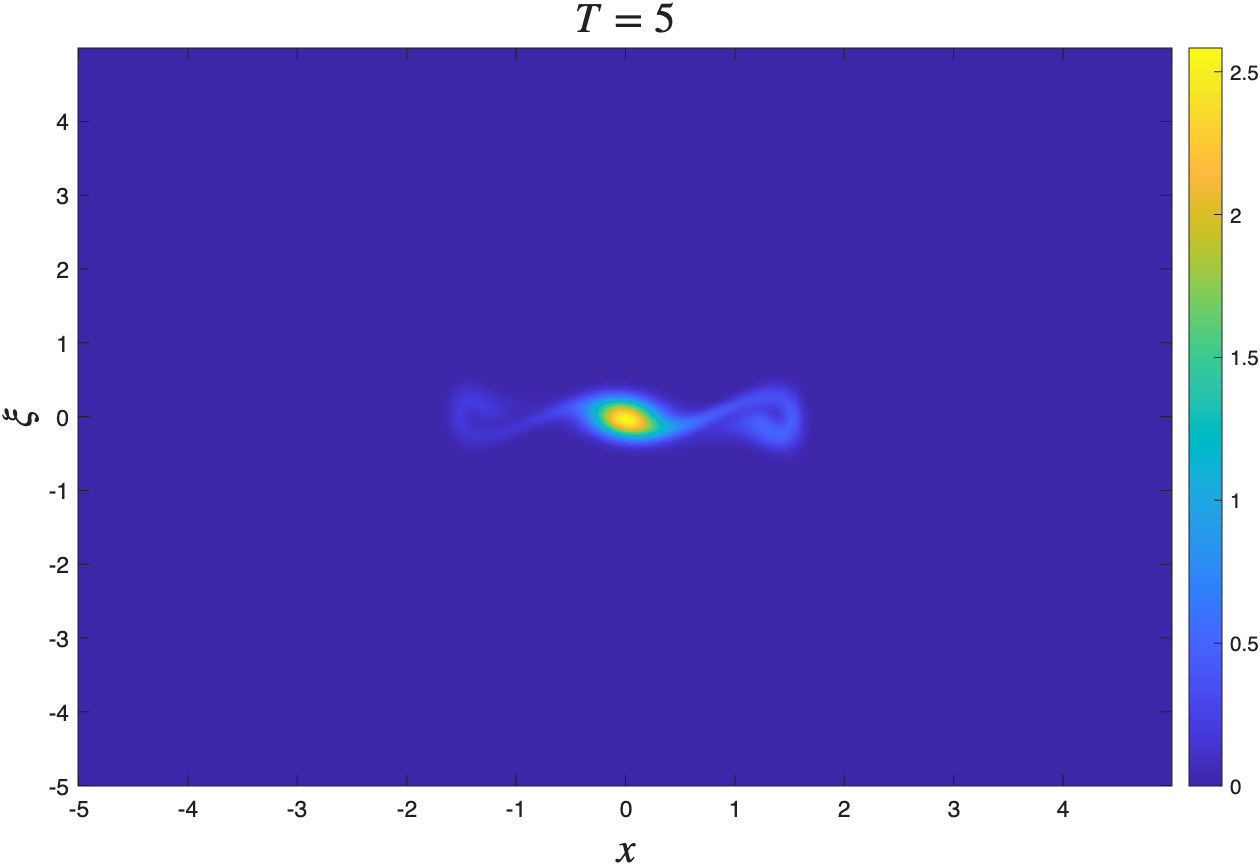}
        \label{fig: time 1 of potential 3 in example 5 }
    \end{subfigure}
    \hfill
    \begin{subfigure}{0.49\textwidth}
        \centering
        \includegraphics[width=\textwidth]{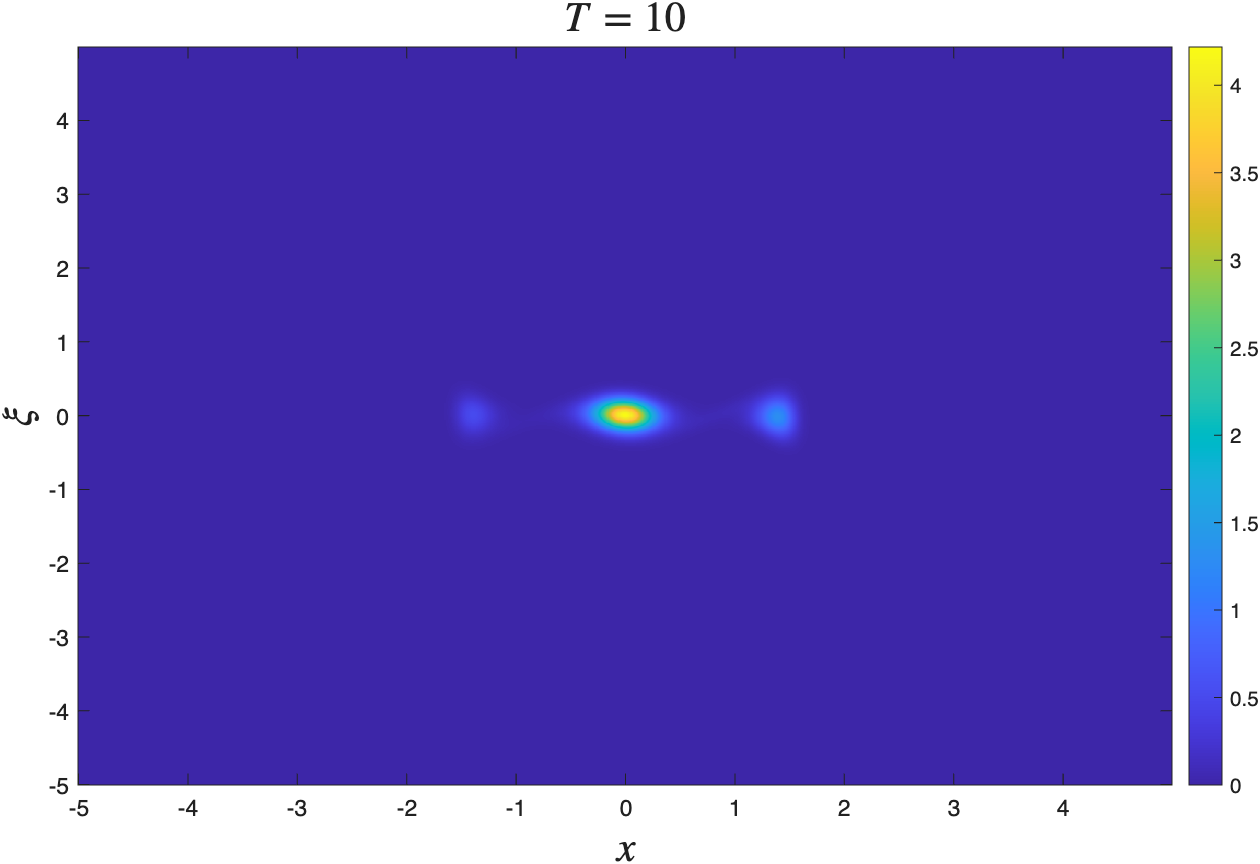}
        \label{fig: time 2 of potential 3 in example 5}
    \end{subfigure}
    \\
    \begin{subfigure}{0.49\textwidth}
        \centering
        \includegraphics[width=\textwidth]{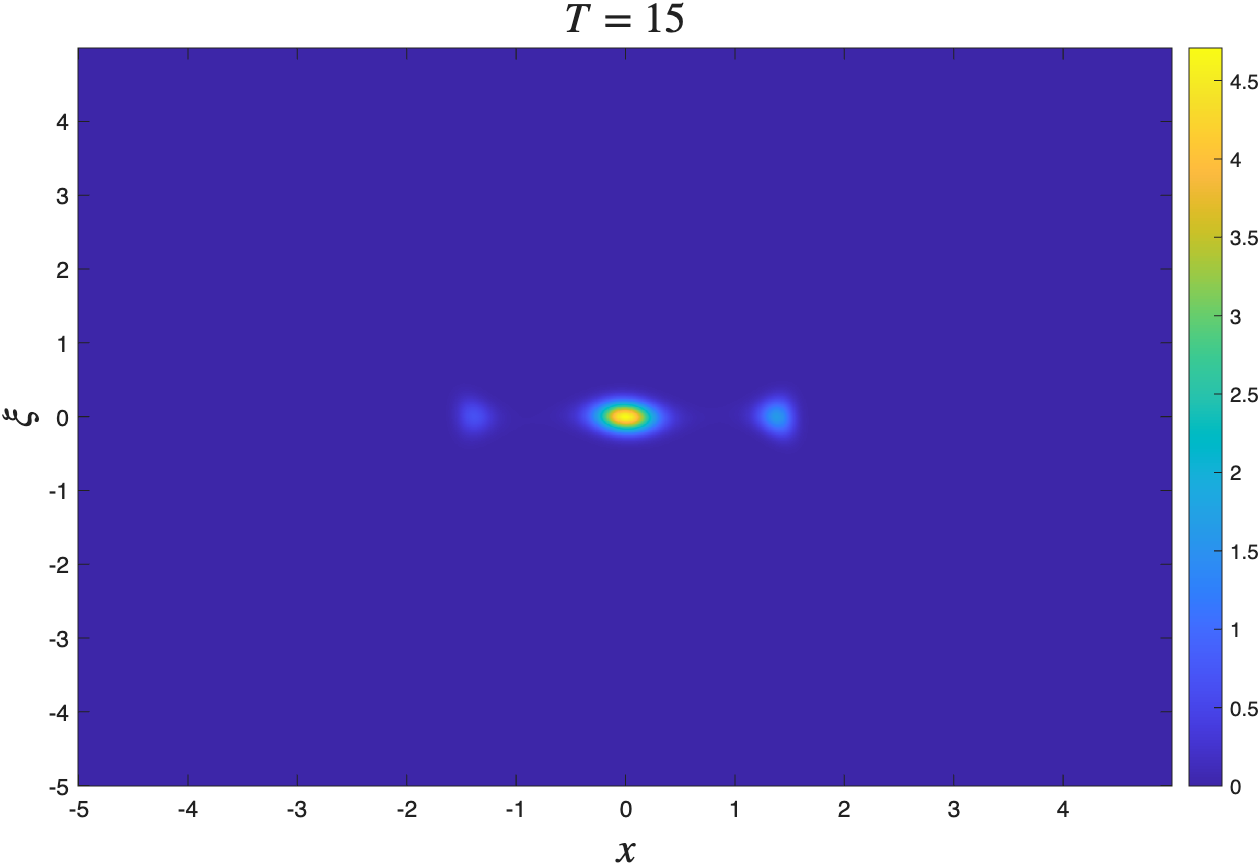}
        \label{fig: time 3 of potential 3 in example 5}
    \end{subfigure}
    \hfill
    \begin{subfigure}{0.49\textwidth}
        \centering
        \includegraphics[width=\textwidth]{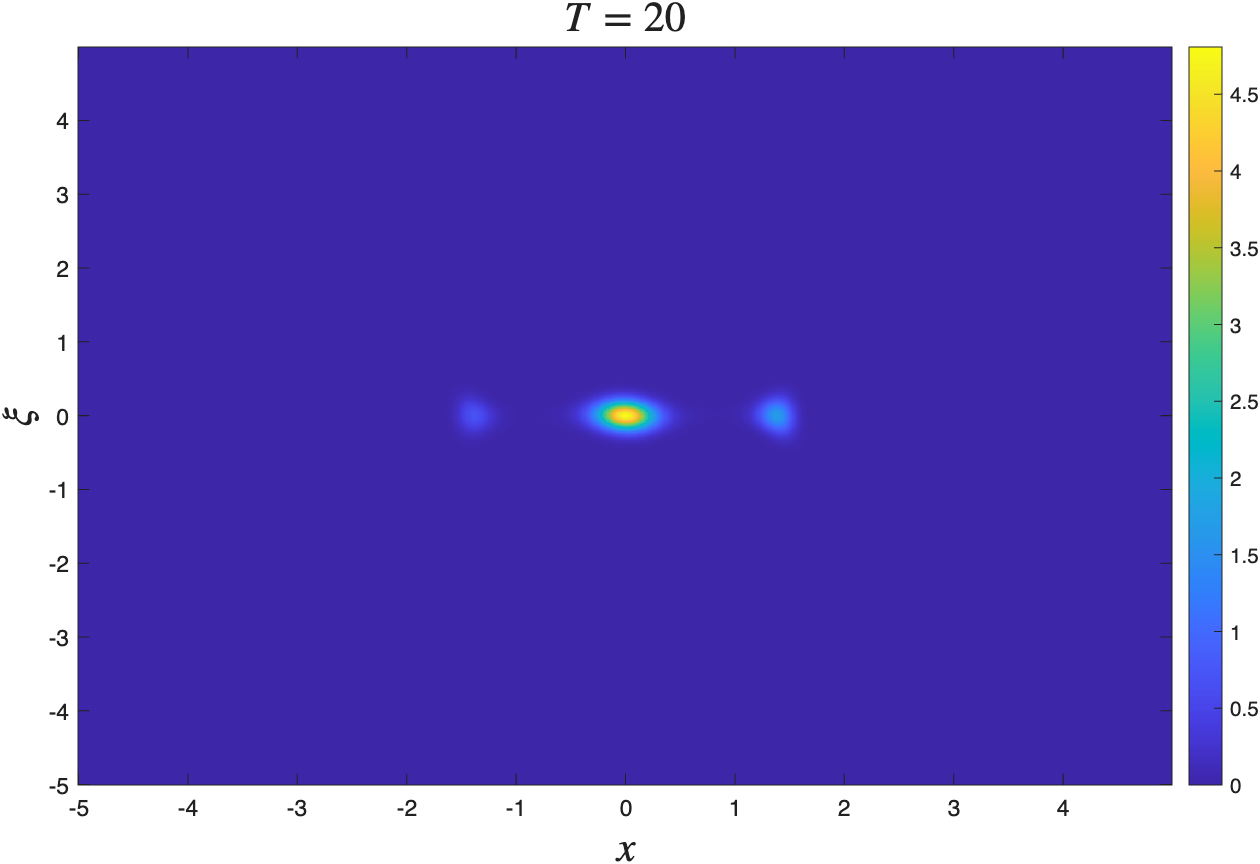}
        \label{fig: time 4 of potential 3 in example 5}
    \end{subfigure}
\caption{Long-time evolution of the Wigner function towards a steady state for Example 4, under the triple-well potential \eqref{pot: potential_3 in example 5} with $\varepsilon=1/32$. From left to right and top to bottom, the plots correspond to times $T=5, 10, 15,$ and $20$. The distribution is visually converged by $T=15$. }
    \label{fig: steady state of potential 3 in example 5}
\end{figure}

\section{Conclusions}
\label{sec:conclusions}

In this article, we have developed and validated a novel Frozen Gaussian Sampling (FGS) algorithm for solving the Wigner-Fokker-Planck equation. Our method demonstrates significant advantages over traditional grid-based approaches, particularly in the challenging semiclassical regime. First, unlike the TSSP method, whose required grid resolution leads to a computational cost that scales unfavorably with $\varepsilon$, the cost of the FGS algorithm depends on the number of samples, which does not need to increase as $\varepsilon$ decreases. This makes the algorithm particularly effective in the semiclassical regime. Second, the mesh-free nature of our algorithm, which evolves Gaussian wave packets in an unbounded phase space, makes it inherently robust for long-time simulations, successfully avoiding the boundary-induced instabilities that constrain fixed-grid methods. Furthermore, we have theoretically proven and numerically verified the algorithm's remarkable robustness when calculating physical observables, for which the sampling error is nearly independent of the semiclassical parameter.

The combination of these advantages positions the FGS algorithm as a powerful new tool for investigating complex phenomena in open quantum systems. In this work, we have applied it to provide strong numerical evidence for the existence of a unique steady state in strongly non-harmonic potentials—a regime where analytical proofs are currently lacking. Our numerical findings pave the way for future theoretical work, particularly in providing a rigorous mathematical proof for the existence of steady states in these general potentials. Future research could extend this method to higher-dimensional systems and explore its application to other challenging problems in quantum dynamics, such as non-adiabatic processes or systems with more complex environmental couplings.

\appendix

\section{Derivation of the Wigner-Fokker-Planck Equation}
\label{appendix:wfp_derivation}
This appendix provides a detailed derivation of the Wigner-Fokker-Planck (WFP) equation. The derivation proceeds by directly using the relationship between the Wigner function $W$ and the density matrix kernel $\rho$ given by the Wigner transform. We start by taking the time derivative of the Wigner function's definition, which allows us to substitute the Lindblad equation governing $\partial_t \rho$ into the transform integral. The full procedure can be summarized in three steps: (1) we express $\partial_t W$ as an integral over $\partial_t \rho$; (2) we substitute the Lindblad equation and rewrite the operators in terms of Wigner coordinates; and (3) we simplify each resulting term individually using integration by parts and Fourier transform properties to arrive at the final WFP equation.

We begin with the definition of the Wigner transform:
\begin{equation*}
    W(\bm{x},\bm{\xi},t) = \frac{1}{(\pi\varepsilon)^n}\int \rho(\bm{x}+\bm{y},\bm{x}-\bm{y},t) e^{-\frac{2\mathrm{i}}{\varepsilon}\bm{y}\cdot\bm{\xi}}\mathrm{d}\bm{y}.
\end{equation*}
Taking the time derivative, we can move the derivative inside the integral:
\begin{equation*}
    \partial_t W(\bm{x},\bm{\xi},t) = \frac{1}{(\pi\varepsilon)^n}\int \partial_t\rho(\bm{x}+\bm{y},\bm{x}-\bm{y},t) e^{-\frac{2\mathrm{i}}{\varepsilon}\bm{y}\cdot\bm{\xi}}\mathrm{d}\bm{y}.
\end{equation*}

The first step is to express the time derivative of the density matrix kernel, $\partial_t\rho$, in a form suitable for the Wigner transform. We start with the Lindblad equation in its coordinate representation, using generic spatial variables $(\bm{z}, \bm{w})$ for the arguments of $\rho$:
\begin{equation}
\begin{aligned}
    \partial_t\rho(\bm{z},\bm{w},t) &= -\frac{\mathrm{i}}{\varepsilon}\left(-\frac{\varepsilon^2}{2}\Delta_{\bm{z}}+V(\bm{z})+\frac{\varepsilon^2}{2}\Delta_{\bm{w}}-V(\bm{w})\right)\rho \\ 
    & \quad -\frac{1}{2\varepsilon}\sum_{k=1}^n \left(\alpha_k (z_k-w_k)^2 - \beta_k (\varepsilon \partial_{z_k}+\varepsilon\partial_{w_k})^2\right) \rho \\
    & \quad + \sum_{k=1}^{n} \left(-\operatorname{Im}(\gamma_k) +\mathrm{i}\left[(\gamma_k z_k\partial_{w_k}-\overline{\gamma}_k w_k\partial_{z_k})+\operatorname{Re}(\gamma_k)(z_k\partial_{z_k}-w_k\partial_{w_k})\right]\right)\rho \\
    & \quad -\sum_{k=1}^n\mu_k(z_k\partial_{z_k}+w_k\partial_{w_k}+1)\rho. 
\end{aligned}
\label{eq:Lindblad_zw_coord}
\end{equation}

Next, we perform a change of variables to the Wigner coordinates: the center coordinate $\bm{x}$ and the difference coordinate $\bm{y}$, which are related to $(\bm{z}, \bm{w})$ by $\bm{z}=\bm{x}+\bm{y}$ and $\bm{w}=\bm{x}-\bm{y}$. To rewrite the differential operators in the $(\bm{x}, \bm{y})$ frame, we apply the chain rule:
\begin{align*}
    \Delta_{\bm{z}}-\Delta_{\bm{w}}&=\nabla_{\bm{y}}\cdot\nabla_{\bm{x}},\quad 
    \left(\varepsilon\partial_{z_k}+\varepsilon\partial_{w_k}\right)^2=\varepsilon^2\partial_{x_k}^2,\quad
    z_k\partial_{z_k}+w_k\partial_{w_k}=x_k\partial_{x_k}+y_k\partial_{y_k},
\end{align*}
\begin{align*}
\left(\gamma_k z_k\partial_{w_k}-\overline{\gamma}_k w_k\partial_{z_k}\right)+\operatorname{Re}\left(\gamma_k\right)\left(z_k\partial_{z_k}-w_k\partial_{w_k}\right)
=2\operatorname{Re}(\gamma_k)y_k\partial_{x_k}+\mathrm{i}\operatorname{Im}(\gamma_k)\left(x_k\partial_{x_k}-y_k\partial_{y_k}\right).
\end{align*}

We now substitute these transformed operators into the integral expression for $\partial_t W$. This yields a single, large expression which, for clarity, we break down into six terms corresponding to the different parts of the Lindblad dynamics. These terms will be analyzed individually in the subsequent steps. The full expression is:
\begin{align*}
\partial_t W(\bm{x},\bm{\xi},t)
&=\frac{1}{(\pi\varepsilon)^n}\int -\frac{\mathrm{i}}{\varepsilon}\nabla_{\bm{y}}\cdot\left[-\frac{\varepsilon^2}{2}\nabla_{\bm{x}}\right]\rho(\bm{x}+\bm{y},\bm{x}-\bm{y},t)\mathrm{e}^{-\frac{\mathrm{i}}{\varepsilon}2\bm{y}\cdot\bm{\xi}}\mathrm{d}\bm{y}\tag{1}\\
&+\frac{1}{(\pi\varepsilon)^n}\int -\frac{\mathrm{i}}{\varepsilon}\left[V(\bm{x}+\bm{y})-V(\bm{x}-\bm{y})\right]\rho(\bm{x}+\bm{y},\bm{x}-\bm{y},t)\mathrm{e}^{-\frac{\mathrm{i}}{\varepsilon}2\bm{y}\cdot\bm{\xi}}\mathrm{d}\bm{y}\tag{2}\\
&+\frac{\varepsilon}{2}\frac{1}{(\pi\varepsilon)^n}\int \left(\sum_{k=1}^{n}\alpha_k\cdot\left(-\frac{\mathrm{i}}{\varepsilon}2y_k\right)^2\right)\rho(\bm{x}+\bm{y},\bm{x}-\bm{y},t)\mathrm{e}^{-\frac{\mathrm{i}}{\varepsilon}2\bm{y}\cdot\bm{\xi}}\mathrm{d}\bm{y}\tag{3}\\
&+\frac{\varepsilon}{2}\frac{1}{(\pi\varepsilon)^n}\int \left(\sum_{k=1}^n\beta_k\partial_{xx}\rho(\bm{x}+\bm{y},\bm{x}-\bm{y},t)\right)\mathrm{e}^{-\frac{\mathrm{i}}{\varepsilon}2\bm{y}\cdot\bm{\xi}}\mathrm{d}\bm{y}\tag{4}\\
&+\frac{1}{(\pi\varepsilon)^n}\int \sum_{k=1}^{n} \left(-\operatorname{Im}(\gamma_k) +\mathrm{i}\left(2\operatorname{Re}(\gamma_k)y_k\partial_{x_k}+\mathrm{i}\operatorname{Im}(\gamma_k)\left(x_k\partial_{x_k}-y_k\partial_{y_k}\right)\right)\right)\rho(\bm{x}+\bm{y},\bm{x}-\bm{y},t)\mathrm{e}^{-\frac{\mathrm{i}}{\varepsilon}2\bm{y}\cdot\bm{\xi}}\mathrm{d}\bm{y}\tag{5}\\
&-\frac{1}{(\pi\varepsilon)^n}\int \sum_{k=1}^{n} \mu_k\left(1 +x_k\partial_{x_k}+y_k\partial_{y_k}\right)\rho(\bm{x}+\bm{y},\bm{x}-\bm{y},t)\mathrm{e}^{-\frac{\mathrm{i}}{\varepsilon}2\bm{y}\cdot\bm{\xi}}\mathrm{d}\bm{y}\tag{6}.
\end{align*}

We now simplify each of the six terms individually.

\paragraph{Term (1): The Kinetic Term}
This term arises from the kinetic energy part of the Hamiltonian. We apply integration by parts with respect to $\bm{y}$, which transfers the $\nabla_{\bm{y}}$ derivative onto the exponential factor:
\begin{align*}
(1)
&=\frac{1}{(\pi\varepsilon)^n}\int -\frac{\mathrm{i}}{\varepsilon}\left(-\frac{\mathrm{i}}{\varepsilon}2\bm{\xi}\right)\cdot\left[\frac{\varepsilon^2}{2}\nabla_{\bm{x}}\right]\rho(\bm{x}+\bm{y},\bm{x}-\bm{y},t)\mathrm{e}^{-\frac{\mathrm{i}}{\varepsilon}2\bm{y}\cdot\bm{\xi}}\mathrm{d}\bm{y}\\
&=-\bm{\xi}\cdot\nabla_{\bm{x}}W(\bm{x},\bm{\xi},t).	
\end{align*}

\paragraph{Term (2): The Potential Term}
This term corresponds to the potential $V(\bm{x})$. By using the Fourier transform:
\begin{align*}
\rho(\bm{x}+\bm{y},\bm{x}-\bm{y})=\int W(\bm{x},\bm{\xi})\mathrm{e}^{\frac{\mathrm{i}}{\varepsilon}2\bm{y}\cdot\bm{\xi}}\mathrm{d}\bm{\xi},	
\end{align*}
we can express the integral as the action of the pseudo-differential operator $\Theta[V]$ on $W$:
\begin{align*}
    (2) &= -\frac{\mathrm{i}}{\varepsilon}\frac{1}{(\pi\varepsilon)^n}\iint\left(V(\bm{x}+\bm{y})-V(\bm{x}-\bm{y})\right)W(\bm{x},\bm{\xi}',t)e^{\frac{2\mathrm{i}}{\varepsilon}\bm{y}\cdot(\bm{\xi}'-\bm{\xi})}\mathrm{d}\bm{\xi}'\mathrm{d}\bm{y} \\
    &=: -\Theta[V]W.
\end{align*}

\paragraph{Term (3) and (4): The Diffusion Terms}
Term (3), containing $\alpha_k$, corresponds to diffusion in momentum. We use the Fourier identity that multiplication by $y_k$ in the $\bm{y}$-space is equivalent to a derivative with respect to $\xi_k$ in the $\bm{\xi}$-space. Specifically, the factor $(2y_k)^2$ becomes $(i\varepsilon \partial_{\xi_k})^2$:
\begin{align*}
    (3) &= \frac{\varepsilon}{2}\sum_{k=1}^n\alpha_k \left( \frac{1}{(\pi\varepsilon)^n}\int \left(-\mathrm{i}\frac{2y_k}{\varepsilon}\right)^2 \rho(\bm{x}+\bm{y},\bm{x}-\bm{y},t)\mathrm{e}^{-\frac{\mathrm{i}}{\varepsilon}2\bm{y}\cdot\bm{\xi}}\mathrm{d}\bm{y}\right) \\
    &= \frac{\varepsilon}{2}\sum_{k=1}^n\alpha_k\partial_{\xi_k\xi_k}W(\bm{x},\bm{\xi},t).
\end{align*}
Term (4), containing $\beta_k$, corresponds to diffusion in position. Since the derivatives are with respect to $\bm{x}$, they can be moved outside the integral:
\begin{align*}
    (4) &= \frac{\varepsilon}{2}\sum_{k=1}^n\beta_k \partial_{x_k x_k} \left( \frac{1}{(\pi\varepsilon)^n}\int \rho(\bm{x}+\bm{y},\bm{x}-\bm{y},t)\mathrm{e}^{-\frac{\mathrm{i}}{\varepsilon}2\bm{y}\cdot\bm{\xi}}\mathrm{d}\bm{y} \right) \\
    &= \frac{\varepsilon}{2}\sum_{k=1}^n\beta_k\partial_{x_kx_k}W(\bm{x},\bm{\xi},t).	
\end{align*}

\paragraph{Term (5) and (6): The Dissipative and Lamb Shift Terms}
Finally, we transform the terms corresponding to dissipation and the Lamb shift. This involves applying similar Fourier identities and integration by parts:
\begin{align*}
    (5) &= -\varepsilon\sum_{k=1}^n\operatorname{Re}(\gamma_k)\partial_{x_k}\partial_{\xi_k}W - \sum_{k=1}^n\operatorname{Im}(\gamma_k)\left(2W+x_k\partial_{x_k}W+\xi_k\partial_{\xi_k}W\right), \\
    (6) &=\sum_{k=1}^n\mu_k\left(-x_k\partial{x_k}W+\xi_k\partial{\xi_k}W\right).
\end{align*}

Finally, we collect all the transformed terms (1) through (6). The left-hand side of the WFP equation is formed by rearranging the $\partial_t W$ term and the results from Term (1) and Term (2). The right-hand side, which constitutes the dissipative and diffusive operator $\mathcal{D}(W)$, is formed by summing the results from Terms (3) through (6). This procedure yields the final Wigner-Fokker-Planck equation:
\begin{equation*}
\begin{aligned}
\partial_t W(\bm{x},\bm{\xi},t)
+\bm{\xi}\cdot\nabla_{\bm{x}}W
+ \Theta[V]W
&=
\frac{\varepsilon}{2}\sum_{k=1}^n\left(\alpha_k\partial_{\xi_k\xi_k}W+\beta_k\partial_{x_kx_k}W- 2\operatorname{Re}(\gamma_k) \partial_{x_k}\partial_{\xi_k}W\right)\\
&\quad - \sum_{k=1}^n \left[ (\operatorname{Im}(\gamma_k) + \mu_k)\partial_{x_k}(x_kW) +(-\operatorname{Im}(\gamma_k)  +\mu_k)\partial_{\xi_k} (\xi_k W) \right].
\end{aligned}
\end{equation*}
This is precisely the WFP equation presented in the main text as \cref{eq:WFP_main}.

\section{Derivation of wavepacket dynamics for the Wigner Fokker-Planck equation}\label{appendix: wavepacket dynamics}
This appendix provides a detailed derivation of the equations of motion for a generic Gaussian wavepacket evolving under the Wigner-Fokker-Planck equation with a locally quadratic potential.

Our starting point is the WFP equation under the Taylor approximation for the potential, as given in \cref{eq:WFP_local_potential}. For a more compact representation, we first rewrite this equation. We define the phase space coordinate $\bm{y}$, its corresponding vector field $\bm{k}$, and the wavepacket center coordinate $\bm{z}$:
$$ \bm{y}=\begin{pmatrix} \bm{x}\\ \bm{\xi} \end{pmatrix}, \quad \boldsymbol{k}(\bm{y},t)= \begin{pmatrix}  \boldsymbol{\xi}\\  -\left(\nabla_{\boldsymbol{x}}V(\bm{q})+\nabla_{\boldsymbol{x}}^2V(\bm{q})(\boldsymbol{x}-\boldsymbol{q})\right) \end{pmatrix}, \quad \bm{z}(t)=\begin{pmatrix} \bm{q}(t)\\ \bm{p}(t) \end{pmatrix}. $$
We also introduce the auxiliary matrices $\Gamma_1, \Gamma_2, \bm{B}, \widetilde{\mathcal{M}}$, which were previously defined in Section \ref{sec:wavepacket_dynamics}.
Using this notation, the WFP equation under the local potential approximation can be expressed in the following compact form:
\begin{align*}
\partial_t W+\bm{k}^{T}\nabla_{\bm{y}}W
&=\bm{y}^T\left(\Gamma_2+\widetilde{\mathcal{M}}\right)\nabla_{\bm{y}}W+\operatorname{Tr}\left(\Gamma_2\right)W\\
&+\frac{\varepsilon}{2}\sum_{k=1}^n\left(\alpha_k\partial_{\xi_k\xi_k}W+\beta_k\partial_{x_kx_k}W- 2\operatorname{Re}(\gamma_k) \partial_{x_k}\partial_{\xi_k}W\right).
\end{align*}

We now substitute the Gaussian ansatz for a generic wavepacket into the compact WFP equation. The ansatz is given by:
\begin{equation}
     W(\bm{x},\bm{\xi},t) = A(t)\exp\left(-\frac{1}{2\varepsilon}T(\bm{x},\bm{\xi},t)\right),
\end{equation}
where for simplicity we have written $T(t)$ for $T(\bm{x}, \bm{\xi}, \bm{q}(t), \bm{p}(t); \bm{\Sigma}(t))$. The time derivative of $W$ is
$$ \partial_t W = \left( \frac{1}{A}\frac{dA}{dt} - \frac{1}{2\varepsilon}\frac{dT}{dt} \right)W. $$
The second-order derivatives required for the diffusive and dissipative terms are found by direct calculation:
\begin{align*}
    \partial_{x_k x_k}W &= W\left[ \frac{1}{4\varepsilon^2}(\partial_{x_k}T)^2 - \frac{1}{2\varepsilon}\partial_{x_k x_k}T \right], \\
    \partial_{\xi_k \xi_k}W &= W\left[ \frac{1}{4\varepsilon^2}(\partial_{\xi_k}T)^2 - \frac{1}{2\varepsilon}\partial_{\xi_k \xi_k}T \right], \\
    \partial_{x_k \xi_k}W &= W\left[ \frac{1}{4\varepsilon^2}(\partial_{x_k}T)(\partial_{\xi_k}T) - \frac{1}{2\varepsilon}\partial_{x_k \xi_k}T \right].
\end{align*}
Substituting all these expressions back into the compact WFP equation and dividing by $W$ yields the following scalar equation for the parameters:
\begin{align*}
&\frac{1}{A}\frac{\mathrm{d}A}{\mathrm{d}t}+\left(-\frac{1}{2\varepsilon}\right)\frac{\mathrm{d}T}{\mathrm{d}t}+\left(-\frac{1}{2\varepsilon}\right)\bm{k}^T\nabla_{\bm{y}}T\\
=&\left(-\frac{1}{2\varepsilon}\right)\bm{y}^T\left(\Gamma_2+\widetilde{\mathcal{M}}\right)\nabla_{\bm{y}}T+\operatorname{Tr}\left(\Gamma_2\right)+\frac{1}{8\varepsilon}\left(\nabla_{\bm{y}}T\right)^T \left(B-\Gamma_1\right)	\nabla_{\bm{y}}T\\
-&\frac{1}{4}\sum_{k=1}^n\left(\beta_k\partial_{x_kx_k}T+\alpha_k\partial_{\xi_k\xi_k}T\right)+\frac{1}{2}\sum_{k=1}^n\operatorname{Re}\left(\gamma_k\right)\partial_{x_k\xi_k}T.
\end{align*}

We now perform an asymptotic analysis by collecting terms of the same order in $\varepsilon$. The leading order terms, $O(\varepsilon^{-1})$, yield an equation for the quadratic form $T$:
\begin{align*}
    O(\varepsilon^{-1}):& \quad\frac{\partial T}{\partial t}+\left(\nabla_{\bm{z}}T\right)^T\frac{\mathrm{d}\bm{z}}{\mathrm{d}t}+\bm{k}^T\nabla_{\bm{y}}T=\bm{y}^T\left(\Gamma_2+\widetilde{\mathcal{M}}\right)\nabla_{\bm{y}}T-\frac{1}{4}\left(\nabla_{\bm{y}}T\right)^T \left(\bm{B}-\Gamma_1\right)\nabla_{\bm{y}}T.
\end{align*}
The next-to-leading order terms, $O(1)$, provide the equation of motion for the amplitude $A$:
\begin{align*}
    O(1):& \quad \frac{\mathrm{d}A}{\mathrm{d}t}=\operatorname{Tr}\left(\Gamma_2-\frac{1}{2}\left(\bm{B}-\Gamma_1\right)\bm{G}\right)A.
\end{align*}
This is the final ODE for the amplitude, \cref{eq:A_dynamics_open}.

The ODEs for the center $\bm{z}$ and the inverse covariance matrix $\bm{G}$ are extracted by taking successive derivatives of the $O(\varepsilon^{-1})$ equation and then evaluating them at the wavepacket center $\bm{y}=\bm{z}$.

First, differentiating the $O(\varepsilon^{-1})$ equation once with respect to $\bm{y}$ gives
\begin{align}\label{eq: first differential with y}
&\frac{\partial\left(\nabla_{\bm{y}}T\right)}{\partial t}
+\nabla_{\bm{y}}\left(\nabla_{\bm{z}}T\right)^T\frac{\mathrm{d}\bm{z}}{\mathrm{d}t}
+\nabla_{\bm{y}}\left(\nabla_{\bm{y}}T\right)^T\bm{k}
+\left(\nabla_{\bm{y}}^T\bm{k}\right)^T\nabla_{\bm{y}}T\\
=&\nabla_{\bm{y}}\left(\bm{y}^T\right)\left(\Gamma_2+\widetilde{\mathcal{M}}\right)\nabla_{\bm{y}}T+\nabla_{\bm{y}}\left(\nabla_{\bm{y}}T\right)^T\left(\Gamma_2+\widetilde{\mathcal{M}}\right)\bm{y}
-\frac{1}{2}\nabla_{\bm{y}}\left(\nabla_{\bm{y}}T\right)^T\left(B-\Gamma_1\right)\nabla_{\bm{y}}T.
\end{align}
We then evaluate this expression at $\bm{y}=\bm{z}$. Using the identities $\nabla_{\bm{y}}T|_{\bm{y}=\bm{z}}=0$, $\nabla_{\bm{y}}(\nabla_{\bm{z}}T)^T|_{\bm{y}=\bm{z}}=-2\bm{G}$, and $\nabla_{\bm{y}}(\nabla_{\bm{y}}T)^T|_{\bm{y}=\bm{z}}=2\bm{G}$, we obtain the equation for the wavepacket center:
\begin{align*}
    \frac{\mathrm{d}\bm{z}}{\mathrm{d}t} = -\left(\Gamma_2+\widetilde{\mathcal{M}}\right)\bm{z}+\bm{k}|_{\bm{y}=\bm{z}}.
\end{align*}
Substituting the explicit forms of $\bm{k}$, $\Gamma_2$, and $\widetilde{\mathcal{M}}$ yields the component-wise equations of motion, \cref{eq:q_dynamics_open} and \cref{eq:p_dynamics_open}.

Next, we differentiate \cref{eq: first differential with y} a second time with respect to $\bm{y}$:
\begin{align*}
&\frac{\partial \nabla_{\bm{y}}^T\left(\nabla_{\bm{y}}T\right)}{\partial t}
+\nabla_{\bm{y}}\left(\nabla_{\bm{y}}T\right)^T\nabla_{\bm{y}}^T\bm{k}
+(\nabla_{\bm{y}}^T\bm{k})^T\nabla_{\bm{y}}^T\left(\nabla_{\bm{y}}T\right)\\
=&\nabla_{\bm{y}}\left(\bm{y}^T\right)\left(\Gamma_2+\widetilde{\mathcal{M}}\right)\nabla_{\bm{y}}^T\left(\nabla_{\bm{y}}T\right)
+\nabla_{\bm{y}}\left(\nabla_{\bm{y}}T\right)^T\left(\Gamma_2+\widetilde{\mathcal{M}}\right)\nabla_{\bm{y}}^T\bm{y}
-\frac{1}{2}\nabla_{\bm{y}}\left(\nabla_{\bm{y}}T\right)^T\left(\bm{B}-\Gamma_1\right)\nabla_{\bm{y}}^T\left(\nabla_{\bm{y}}T\right).	
\end{align*}
Again, evaluating at $\bm{y}=\bm{z}$ and using the identity $\nabla_{\bm{y}}^T\bm{k}|_{\bm{y}=\bm{z}} = \bm{C}(t)$, we arrive at the equation of motion for the inverse covariance matrix, \cref{eq:G_dynamics_open}:
\begin{align*}
    \frac{\mathrm{d} \bm{G}}{\mathrm{d} t}=\left(\Gamma_2+\widetilde{\mathcal{M}}-\bm{C}^T\right)\bm{G}+\bm{G}\left(\Gamma_2+\widetilde{\mathcal{M}}-\bm{C}\right)+\bm{G}\left(\Gamma_1-\bm{B}\right)\bm{G}.
\end{align*}
This completes the derivation of the wavepacket dynamics presented in Section 3.

\section*{Acknowledgments}
The work of Z. Zhou was partially supported by the National Key R\&D Program of China (Project No. 2021YFA1001200), and the National Natural Science Foundation of China (Grant No. 12171013). We thank Westlake University HPC Center for computation support. 

\bibliographystyle{elsarticle-num}
\bibliography{FGS_lindblad}

@article{diehl2008quantum,
  title={Quantum states and phases in driven open quantum systems with cold atoms},
  author={Diehl, Sebastian and Micheli, A and Kantian, Adrian and Kraus, B and B{\"u}chler, HP and Zoller, Peter},
  journal={Nature Physics},
  volume={4},
  number={11},
  pages={878--883},
  year={2008},
  publisher={Nature Publishing Group UK London}
}

@book{breuer2002theory,
  title={The theory of open quantum systems},
  author={Breuer, Heinz-Peter and Petruccione, Francesco},
  year={2002},
  publisher={Oxford University Press, USA}
}

@article{lindblad1976generators,
  title={On the generators of quantum dynamical semigroups},
  author={Lindblad, Goran},
  journal={Communications in Mathematical Physics},
  volume={48},
  pages={119--130},
  year={1976},
  publisher={Springer}
}

@article{wang2020combining,
  title={Combining {Lindblad} master equation and surface hopping to evolve distributions of quantum particles},
  author={Wang, Yi-Siang and Nijjar, Parmeet and Zhou, Xin and Bondar, Denys I and Prezhdo, Oleg V},
  journal={The Journal of Physical Chemistry B},
  volume={124},
  number={21},
  pages={4326--4337},
  year={2020},
  publisher={ACS Publications}
}

@article{cao2017lindblad,
  title={{Lindblad} equation and its semiclassical limit of the {Anderson-Holstein} model},
  author={Cao, Yu and Lu, Jianfeng},
  journal={Journal of Mathematical Physics},
  volume={58},
  number={12},
  year={2017},
  publisher={AIP Publishing}
}

@article{buvca2012note,
  title={A note on symmetry reductions of the {Lindblad} equation: transport in constrained open spin chains},
  author={Bu{\v{c}}a, Berislav and Prosen, Toma{\v{z}}},
  journal={New Journal of Physics},
  volume={14},
  number={7},
  pages={073007},
  year={2012},
  publisher={IOP Publishing}
}

@article{lotem2020renormalized,
  title={Renormalized {Lindblad} driving: A numerically exact nonequilibrium quantum impurity solver},
  author={Lotem, Matan and Weichselbaum, Andreas and von Delft, Jan and Goldstein, Moshe},
  journal={Physical Review Research},
  volume={2},
  number={4},
  pages={043052},
  year={2020},
  publisher={APS}
}

@article{chen2023quantum,
  title={Quantum thermal state preparation},
  author={Chen, Chi-Fang and Kastoryano, Michael J and Brand{\~a}o, Fernando GSL and Gily{\'e}n, Andr{\'a}s},
  journal={arXiv preprint arXiv:2303.18224},
  year={2023}
}

@article{chen2025randomized,
  title={A randomized method for simulating Lindblad equations and thermal state preparation},
  author={Chen, Hongrui and Li, Bowen and Lu, Jianfeng and Ying, Lexing},
  journal={Quantum},
  volume={9},
  pages={1917},
  year={2025},
  publisher={Verein zur F{\"o}rderung des Open Access Publizierens in den Quantenwissenschaften}
}

@article{dubois2021semi,
  title={Semi-classical {Lindblad} master equation for spin dynamics},
  author={Dubois, Jonathan and Saalmann, Ulf and Rost, Jan M},
  journal={Journal of Physics A: Mathematical and Theoretical},
  volume={54},
  number={23},
  pages={235201},
  year={2021},
  publisher={IOP Publishing}
}

@article{graefe2018lindblad,
  title={{Lindblad} dynamics of {Gaussian} states and their superpositions in the semiclassical limit},
  author={Graefe, EM and Longstaff, B and Plastow, T and Schubert, R},
  journal={Journal of Physics A: Mathematical and Theoretical},
  volume={51},
  number={36},
  pages={365203},
  year={2018},
  publisher={IOP Publishing}
}

@article{galkowski2025classical,
  title={Classical--quantum correspondence in Lindblad evolution},
  author={Galkowski, Jeffrey and Zworski, Maciej and Huang, Zhen},
  journal={Journal of Mathematical Physics},
  volume={66},
  number={9},
  year={2025},
  publisher={AIP Publishing}
}

@article{hernandez2025classical,
  title={Classical correspondence beyond the {Ehrenfest} time for open quantum systems with general {Lindbladians}},
  author={Hern{\'a}ndez, Felipe and Ranard, Daniel and Riedel, C Jess},
  journal={Communications in Mathematical Physics},
  volume={406},
  number={1},
  pages={1--81},
  year={2025},
  publisher={Springer}
}

@article{li2024long,
  title={Long time quantum-classical correspondence for open systems in trace norm},
  author={Li, Zhenhao},
  journal={arXiv preprint arXiv:2408.16953},
  year={2024}
}

@article{lu2012convergence,
  title={Convergence of frozen {Gaussian} approximation for high-frequency wave propagation},
  author={Lu, Jianfeng and Yang, Xu},
  journal={Communications on Pure and Applied Mathematics},
  volume={65},
  number={6},
  pages={759--789},
  year={2012},
  publisher={Wiley Online Library}
}

@article{lu2018frozen,
  title={Frozen {Gaussian} approximation with surface hopping for mixed quantum-classical dynamics: A mathematical justification of fewest switches surface hopping algorithms},
  author={Lu, Jianfeng and Zhou, Zhennan},
  journal={Mathematics of Computation},
  volume={87},
  number={313},
  pages={2189--2232},
  year={2018}
}

@article{huang2023efficient,
  title={Efficient {Frozen Gaussian Sampling} algorithms for nonadiabatic quantum dynamics at metal surfaces},
  author={Huang, Zhen and Xu, Limin and Zhou, Zhennan},
  journal={Journal of Computational Physics},
  volume={474},
  pages={111771},
  year={2023},
  publisher={Elsevier}
}

@article{xie2024frozen,
  title={{Frozen Gaussian Sampling}: A Mesh-free {Monte Carlo} Method For Approximating Semiclassical {Schr{\"o}dinger} Equations},
  author={Xie, Yantong and Zhou, Zhennan},
  journal={Communications in Mathematical Sciences},
  volume={22},
  number={4},
  pages={1133--1166},
  year={2024},
  publisher={International Press of Boston}
}

@article{manzano2020short,
  title={A short introduction to the {Lindblad} master equation},
  author={Manzano, Daniel},
  journal={AIP Advances},
  volume={10},
  number={2},
  year={2020},
  publisher={AIP Publishing}
}

@article{arnold2012wigner,
  title={The {Wigner--Fokker--Planck} equation: stationary states and large time behavior},
  author={Arnold, Anton and Gamba, Irene M and Gualdani, Maria Pia and Mischler, St{\'e}phane and Mouhot, Cl{\'e}ment and Sparber, Christof},
  journal={Mathematical Models and Methods in Applied Sciences},
  volume={22},
  number={11},
  pages={1250034},
  year={2012},
  publisher={World Scientific}
}

@article{arnold2004analysis,
  title={An analysis of quantum {Fokker--Planck} models: A {Wigner} function approach},
  author={Arnold, Anton and L{\'o}pez, Jos{\'e} L and Markowich, Peter A and Soler, Juan},
  journal={Revista Matem{\'a}tica Iberoamericana},
  volume={20},
  number={3},
  pages={771--814},
  year={2004}
}

@inproceedings{arnold2007wigner,
  title={The {Wigner--Poisson--Fokker--Planck} system: global-in-time solution and dispersive effects},
  author={Arnold, Anton and Dhamo, Elidon and Manzini, Chiara},
  booktitle={Annales de l'Institut Henri Poincar{\'e} C, Analyse non lin{\'e}aire},
  volume={24},
  number={4},
  pages={645--676},
  year={2007},
  organization={Elsevier}
}

@article{dieci1994positive,
  title={Positive definiteness in the numerical solution of {Riccati} differential equations},
  author={Dieci, Luca and Eirola, Timo},
  journal={Numerische Mathematik},
  volume={67},
  pages={303--313},
  year={1994},
  publisher={Springer}
}

@article{gorini1976completely,
  title={Completely positive dynamical semigroups of {N}-level systems},
  author={Gorini, Vittorio and Kossakowski, Andrzej and Sudarshan, Ennackal Chandy George},
  journal={Journal of Mathematical Physics},
  volume={17},
  number={5},
  pages={821--825},
  year={1976},
  publisher={American Institute of Physics}
}

@article{golse2021convergence,
  title={On the convergence of time splitting methods for quantum dynamics in the semiclassical regime},
  author={Golse, Fran{\c{c}}ois and Jin, Shi and Paul, Thierry},
  journal={Foundations of Computational Mathematics},
  volume={21},
  number={3},
  pages={613--647},
  year={2021},
  publisher={Springer}
}

@article{chai2024frozen,
  title={{Frozen Gaussian} sampling for scalar wave equations},
  author={Chai, Lihui and Feng, Ye and Zhou, Zhennan},
  journal={ESAIM: Mathematical Modelling and Numerical Analysis},
  volume={58},
  number={5},
  pages={1615--1649},
  year={2024},
  publisher={EDP Sciences}
}

@article{bao2002time,
  title={On time-splitting spectral approximations for the {Schr{\"o}dinger} equation in the semiclassical regime},
  author={Bao, Weizhu and Jin, Shi and Markowich, Peter A},
  journal={Journal of Computational Physics},
  volume={175},
  number={2},
  pages={487--524},
  year={2002},
  publisher={Elsevier}
}

@article{sparber2004long,
  title={On the long-time behavior of the quantum {Fokker-Planck} equation},
  author={Sparber, Christof and Carrillo, Jos{\'e} A and Dolbeault, Jean and Markowich, Peter A},
  journal={Monatshefte f{\"u}r Mathematik},
  volume={141},
  number={3},
  pages={237--257},
  year={2004},
  publisher={Springer}
}

@article{yi2025time,
  title={A time-splitting {Fourier} pseudospectral method for the {Wigner (-Poisson)-Fokker-Planck} equations},
  author={Yi, Qian and Xu, Limin},
  journal={arXiv preprint arXiv:2509.11153},
  year={2025}
}

\end{document}